\documentclass[11pt,reqno]{amsart}
\usepackage[margin=1in]{geometry}
\usepackage{graphicx}
\usepackage{amssymb}
\usepackage{amsmath}
\usepackage{amsthm}
\usepackage{comment}
\usepackage{amsrefs}
\usepackage{physics}
\usepackage[shortlabels]{enumitem}
\usepackage{hyperref}
\usepackage{mathrsfs}

\usepackage[textsize=scriptsize,backgroundcolor=white]{todonotes}

\newtheorem{thm}{Theorem}

\newtheorem{prop}[thm]{Proposition}
\newtheorem{lem}[thm]{Lemma}

\theoremstyle{definition}
\newtheorem{defn}[thm]{Definition}
\newtheorem{exm}[thm]{Example}

\theoremstyle{remark}
\newtheorem{rk}[thm]{Remark}

\newcommand{\vertiii}[1]{{\left\vert\kern-0.25ex\left\vert\kern-0.25ex\left\vert #1\right\vert\kern-0.25ex\right\vert\kern-0.25ex\right\vert}}

\newcommand{\ppfxHil}{\mathbb{H}}
\newcommand{\ppHil}{H}
\newcommand{\ppfxCov}{\mathbb{K}}
\newcommand{\ppBall}{\mathbb{B}}
\newcommand{\ppAlg}{\mathcal{A}}
\newcommand{\ppKer}{\mathcal{K}}
\newcommand{\ppWFunc}{\Lambda}
\newcommand{\ppCone}{\mathcal{C}}

\newcommand{\Hol}[0]{ \mathfrak{Hol}}

\begin{document}

\title{On Algebras of Functions over Infinite Dimensions}

\author[D.\ Giannakis]{Dimitrios Giannakis}
\author[M.~J.~Latifi Jebelli]{Mohammad Javad Latifi Jebelli}
\author[M.\ Montgomery]{Michael Montgomery}
\address{Department of Mathematics, Dartmouth College, Hanover, New Hampshire, USA}
\email{mohammad.javad.latifi.jebelli@dartmouth.edu}

\maketitle

\begin{abstract}
We introduce a family of reproducing kernel Hilbert spaces $\mathcal A_\Lambda$ of holomorphic functions defined on an infinite--dimensional domain in a separable Hilbert space, $\ppfxHil$. The reproducing kernel of $\mathcal A_\Lambda$ is constructed using the covariance operator associated with a Gaussian measure on $\ppfxHil$, along with a holomorphic function $\Lambda$ on the unit disk. 
Under certain conditions on the kernel, $\mathcal A_\Lambda$ is closed under pointwise multiplication, giving it the structure of a reproducing kernel Hilbert algebra (RKHA). 
We also study twisted canonical commutation relations on these RKHAs, where the creation and annihilation operators are both bounded.
\end{abstract}

\section{Introduction}

We construct and study a family of function spaces on an infinite-dimensional Hilbert space that are simultaneously reproducing kernel Hilbert spaces (RKHSs) and Banach algebras with respect to pointwise multiplication. 

\subsection*{Cone in the Wiener algebra} Our construction of reproducing kernels depends on the choice of a function in the `positive cone' $\ppCone^{W}$  of the Wiener algebra of complex-valued functions on the unit circle with absolutely summable Fourier coefficients. 
Let $\Delta = \{ z\in \mathbb{C}: |z| < 1 \}$. The holomorphic part of the Wiener algebra is given by $W(\Delta) = \{ \sum_{n=0}^{\infty} \lambda_n z^n: \sum_{n=0}^{\infty} |\lambda_n|<\infty \}$. In other words, $W(\Delta)$ consists of holomorphic functions on $\Delta$ with absolutely summable coefficients $\lambda_n$ (where absolute summability guarantees uniform convergence on $\bar \Delta$). We define $ \ppCone^W$ to be the set of all functions $\sum_{n=0}^{\infty} \lambda_n z^n$ in $W(\Delta)$ with $\lambda_n \geq 0$, $n=0,1,\dots$. The subset $ \ppCone^W$ defines a \textit{positive cone} since $\ppCone^W+\ppCone^W\subset \ppCone^W$, $r \ppCone^W \subset \ppCone^W$ for any $r\geq 0$, and $\ppCone^W \cap -\ppCone^W = \{0\}$. Let $\leq$ be the order structure on $W(\Delta)$ induced from the positive cone $\ppCone^W$. We define an element $a\in \ppCone^W$ to be \textit{essential} if $a^2 \leq C a$ for some positive constant $C$. We denote by $\ppCone_{ess}^W \subset \ppCone^W$ the set of all essential elements of $\ppCone^W$. 

\subsection*{Reproducing kernel Hilbert space} Let $\ppfxHil$ be an infinite-dimensional, separable Hilbert space. We fix a positive-definite trace class operator $\ppfxCov \colon \ppfxHil \to \ppfxHil$ on $\ppfxHil$. The set
\begin{equation}\label{eq:def-B}
    \ppBall = \{ \xi\in \ppfxHil : \|\ppfxCov^{-1/2}\xi \| < 1 \} \subset \ppfxHil
\end{equation}
has compact closure, $\bar\ppBall$, with respect to the norm topology of $\mathbb H$.
 
Throughout the paper $\ppfxHil$ refers to the same fixed Hilbert space unless stated otherwise, and $ \ppBall$ and $\ppfxCov$ are fixed. We also assume that $\xi_1, \xi_2, \dots $ is an orthonormal eigenbasis for $\ppfxCov$ with eigenvalues $k_1^2, k_2^2, \dots$, respectively.

Given $\ppWFunc \in \ppCone^W$, define $\ppKer_{\ppWFunc}\colon \bar\ppBall \times \bar\ppBall \rightarrow \mathbb{C}$ to be the kernel function
\begin{equation*}
    \ppKer_{\ppWFunc}(\xi,\eta) = \ppWFunc (\langle \xi, \ppfxCov^{-1} \eta \rangle ),
\end{equation*}
where $\langle \cdot, \cdot \rangle$ is the inner product of $\ppfxHil$.
Later we will see that $\mathcal K_\Lambda$ is continuous and the expansion $\ppKer_{\ppWFunc}(\xi,\eta) = \sum_{n=0}^{\infty} \lambda_n \langle \xi, \ppfxCov^{-1} \eta \rangle^n$ is uniformly convergent for $\xi,\eta \in \bar\ppBall$. One also readily verifies that $\mathcal K_\Lambda$ is positive-definite and thus induces an RKHS, $\mathcal A_\Lambda$, of (continuous) complex-valued functions on $\bar\ppBall$.

The family $\{\mathcal K_\Lambda \}_{\Lambda \in \mathcal C^W}$ contains all holomorphic 
kernels that are invariant under the unitary group with respect to the inner product $\langle \xi , \ppfxCov^{-1} \eta \rangle$. These are analogues of RKHSs on finite-dimensional domains with zonal reproducing kernels; e.g., \cite{Gneiting13}. 

\subsection*{Gaussian fields}
 In the absence of Lebesgue measures, Gaussian measures are natural measures in the infinite--dimensional setting. Let $L^2(\mathbb H, \mu)$ be the space of square-integrable functions on $\mathbb{H}$ with respect to a background Gaussian measure $\mu$ on $\mathbb{H}$. We assume that $\ppfxCov \colon \ppfxHil \to \ppfxHil$ is the trace class covariance operator associated with $\mu$. 
 
Gaussian measures play a crucial rule in the theory of stochastic processes. The Cameron--Martin Hilbert space $\ppfxHil_{\mu} \subset \ppfxHil$ is defined via the covariance structure of the Gaussian measure $\mu$ and the pair $(\mathbb{H}_{\mu}, \mathbb{H})$ is called an abstract Wiener space. The interested reader is referred to \cite{Bogachev1998} for Gaussian measures on infinite dimensions, Cameron--Martin Hilbert spaces, and their application in stochastic processes. The space $L^2(\mathbb H, \mu)$ contains as a closed subspace the infinite-dimensional Segal--Bargmann space, $\mathcal{H}L^2(\mathbb{H}, \mu)$, consisting of holomorphic functions of infinitely many variables associated with eigenvectors of $\ppfxCov$. Infinite-dimensional Segal--Bargmann spaces have been studied extensively by many authors; see \cite{MR1770752} for an overview of the subject. 

Spaces of square-integrable functions on an infinite-dimensional background, $\mathbb{H}$, also arise naturally in rigorous treatments of quantum field theory. Associated to a Gaussian measure $\mu$ on $\mathbb{H}$ one defines a Gaussian field $\Phi: \mathbb{H}_{\mu} \rightarrow L^2(\mathbb{H}, \mu)$ as a family of Gaussian random variables in $L^2(\mathbb{H}, \mu)$ indexed by vectors in $\mathbb{H}_{\mu} \subset \mathbb{H}$ with the property that $\langle f, g \rangle_{\mathbb{H}_{\mu}} = \langle \Phi(f), \Phi(g) \rangle_{L^2(\mathbb{H}, \mu)}$. Here, the map $\Phi$ is defined by $\Phi (h) = \langle h, \cdot \rangle_{\mathbb{H}} \in L^2(\mathbb{H}, \mu)$. In the context of Euclidean quantum field theory, Gaussian fields are known as free fields and they describe systems without interactions. Interacting systems are modeled as non-Gaussian fields $\Phi$ through various regularization techniques. In some cases, the regularization process can be performed in a mathematically rigorous manner. See \cite{Simon1974} for a treatment of Gaussian processes in constructive quantum field theory. 

Given the inclusion $\ppfxHil_{\mu} \subset \ppfxHil$, the set $\ppBall$, defined in (\ref{eq:def-B}), can be written as the unit ball in the Cameron-Martin Hilbert space. i.e. $\ppBall = \{ \xi \in \ppfxHil_{\mu}: \| \xi \|_{\ppfxHil_{\mu}} < 1 \}$. As a result, a function $f \in \ppAlg_{\ppWFunc}$ can be considered as a function on the unit ball of $\ppfxHil_{\mu}$.  Below, we show that every such RKHS $\mathcal A_\Lambda$ consists of holomorphic functions on this unit ball $\mathbb B$. Furthermore, in section \ref{sec:integral-rep},  we show existence of a Gelfand triple $\hat \ppAlg_{\Lambda} \subset \mathcal{H}L^2 (\ppfxHil, \mu) \subset  \ppAlg_{\Lambda}$. Here $\mathcal{H}L^2 (\ppfxHil, \mu)$ is the holomorphic subspace of $L^2 (\ppfxHil, \mu)$, and the inclusion $\mathcal{H}L^2 (\ppfxHil, \mu) \subset  \ppAlg_{\Lambda}$ is given by restricting a holomorphic function on $\ppfxHil$ to $\bar{\ppBall}$. 
In the context of stochastic processes, such a triple defines what are known as generalized Brownian functionals, where elements of $\ppAlg_{\Lambda}$ are the analogs of generalized functions. This traces back to work of Hida \cite{Hida1976ANALYSISOB} and has applications in stochastic partial differential equation theory \cite{Kuo1983BrownianFA,STREIT198455}.

\subsection*{Reproducing kernel Hilbert algebras}
In \cite{GiannakisMontgomery25}, the notion of a reproducing kernel Hilbert algebra (RKHA) was defined as an RKHS with additional coalgebra structure.

\begin{defn}
    \label{def:rkha} An RKHS $\ppAlg$ on a set $X$ with reproducing kernel $k \colon X \times X \to \mathbb C$ is an RKHA if it admits a bounded comultiplication operator $\Delta \colon \ppAlg \to \ppAlg \otimes \ppAlg$ satisfying $\Delta k(x,\cdot) = k(x,\cdot) \otimes k(x, \cdot)$ for all $x \in X$. 
\end{defn}

One verifies that the adjoint of $\Delta$ from Definition~\ref{def:rkha} is a (bounded) operator that implements pointwise function multiplication, $\Delta^*(f \otimes g) = fg$. This makes $\ppAlg$ a Banach algebra with respect to a norm induced by the operator norm of multiplication operators on $\ppAlg$; see \cite[Section~2]{GiannakisMontgomery25}.

\subsection*{RKHAs of holomorphic functions}
In this paper, we build RKHAs of holomorphic functions in the family $\mathcal A_\Lambda$. We prove that if $\ppWFunc \in \ppCone^W_{ess}$ then  $\ppAlg_{\ppWFunc}$ is an RKHA of holomorphic functions.

\begin{thm}\label{thm:main}
    Let $\ppWFunc \in \ppCone^W$. Then $\mathcal A_\Lambda$ is an RKHS of functions on the set $\bar \ppBall$ such that:
    \begin{enumerate}[(a)]
        \item Elements of $\mathcal A_\Lambda$ are holomorphic on $\ppBall \subset \mathbb H_\mu$.
        \item $\mathcal A_\Lambda$ is invariant  under the action of the unitary group of $\ppfxHil_{\mu}$.
        \item If $\Lambda \in \ppCone^W_{ess}$ then $\mathcal A_\Lambda$ is a unital RKHA. 
    \end{enumerate}
\end{thm}

For $\tau >0$ and $0<p<1$,  an example of such an algebra is given via 
\begin{gather*}
    \ppWFunc_{\tau,p}(z) = \sum_{n=0}^{\infty} e^{-\tau n^p} z^n, \\
    \ppKer_{\tau,p}(\xi,\eta) =  \ppWFunc_{\tau,p}(\langle \xi, \ppfxCov^{-1} \eta \rangle) ,
 \end{gather*}
where $ \ppWFunc_{\tau,p} \in \ppCone^W_{ess}$ can be deduced from the fact that $\lambda(n)=e^{-\tau n^p}$ is subconvolutive on $\mathbb N$, i.e., $\lambda * \lambda \leq C \lambda$ (see ~\cites{Feichtinger79,Grochenig07}).  Note that the subconvolutive property breaks when $\tau=0$ or $p=1$ but the corresponding spaces $\mathcal A_{\Lambda_{\tau, p}}$ are still RKHSs, with reproducing kernels 
\begin{gather}
    \label{eq:kerneltau0} \ppKer_{0,p}(\xi,\eta) = \frac{1}{1 - \langle \xi, \ppfxCov^{-1} \eta \rangle}, \quad \ppWFunc_{0,p}(z) = \frac{1}{1-z}, \\
    \nonumber \ppKer_{\tau,1}(\xi,\eta) = \frac{1}{1 - e^{\tau}\langle \xi, \ppfxCov^{-1} \eta \rangle}, \quad \ppWFunc_{\tau,1} (z) = \frac{1}{1-e^{-\tau}z}.
\end{gather}
In the case of $\ppfxHil = \mathbb{C}^n$ and $\ppfxCov = I$ the kernel in (\ref{eq:kerneltau0}) reduces to $\ppKer(x,y)=1/(1-x\cdot y)$ and hence $\ppAlg_{\ppWFunc_{0,p}}$ is the Drury--Arveson space on $\mathbb{C}^n$ \cite{Arveson72} (not an algebra anymore), which is the natural generalization of the Hardy space on the disk. From this perspective, for an infinite dimensional $\ppfxHil$, $\ppAlg_{\ppWFunc_{0,p}}$ is the infinite dimensional Drury--Arveson space. The class of RKHSs of holomorphic functions considered here appears in a variety of contexts. For example, \cite{hartz2025weighted} studies the existence of weighted precomposition operators on such spaces.

A natural question arising from Theorem~\ref{thm:main} is the following. What are the necessary and sufficient conditions characterizing all RKHSs of functions on $\ppBall$ that are invariant under the action of the unitary group of $\ppfxHil_{\mu}$ and closed under pointwise multiplication?

In what follows, we give a characterization of the algebra $\ppAlg_{\ppWFunc}$ using weights relative to the Wiener chaos expansion of  $L^2(\ppfxHil,\mu)$ (equivalently, a weighted Fock space). Restricting functions of $L^2(\ppfxHil,\mu)$ to $\bar \ppBall$, both $\ppAlg_{\ppWFunc}$ and $L^2(\ppfxHil,\mu)$ are Hilbert spaces of functions over $\bar \ppBall$ and they share an  orthogonal basis. 
 In related work, \cite{GiannakisEtAl25} put forward a procedure for building RKHAs from a generic Hilbert space, $H$, using a weighted symmetric Fock space construction. Given a weight $w\colon \mathbb N \to \mathbb R_{>0}$ \cite{GiannakisEtAl25} defines $\mathbf{F}_w(H)$ as the closure of the symmetric tensor algebra of $H$ with respect to the inner product $\langle \cdot, \cdot \rangle$ satisfying $\langle \xi_1 \vee \cdots \vee \xi_n, \eta_1 \vee \cdots\vee \eta_n \rangle = w(n)^{-2} \prod_{i=1}^n \langle \xi_i, \eta_i \rangle$ for $\xi_i, \eta_i \in H$.  It is shown \cite[Theorem~9]{GiannakisEtAl25} that if $w^{-2}$ is summable and subconvolutive, $w^{-2} * w^{-2} \leq C w^{-2}$ for a constant $C$, $\mathbf{F}_w(H)$, is a unital Banach algebra with respect to the symmetric tensor product, $\vee$, and a norm equivalent to the Hilbert space norm.  By Gelfand duality, $\mathbf{F}_w(H)$ is isomorphic to a space of continuous functions on its character spectrum, $\sigma(\mathbf{F}_w(H)) \cong \{ \langle \cdot, \xi \rangle: \lVert \xi \rVert \leq R_w\}  \subset H^*$, under the weak-$^*$ topology, where $R_w \in \mathbb{R}_+$. This function space, denoted by $\hat{\mathbf{F}}_w(H)$, turns out to be a unital RKHA.

For compatibility with the algebras $\mathcal A_\Lambda$, our convention for weighted Fock spaces in this paper, denoted by $F_w(H)$, differs from that of \cite{GiannakisEtAl25} by an additional factor. In section~\ref{sec:proof} below we show that the RKHA $\mathcal A_\Lambda$ is isomorphic to the weighted Fock space $ F_w(\mathbb H_\mu^*)$ generated by a space of distributions, $\mathbb H^*_\mu$, on the Cameron--Martin Hilbert space $\mathbb H_\mu \subset \mathbb H$, for the Gaussian measure $\mu$ associated with $\mathbb K$ and a weight $w$ induced from $\Lambda$. As a result, we realize $ F_w(\mathbb H_\mu^*)$, which is a space of continuous functions on the weak-$^*$ compact set $\sigma( F_w(\mathbb H_\mu^*)) \subset F_w(\mathbb H_\mu^*)^*$, as a space of holomorphic functions on the open set $\mathbb B \subset \mathbb H_{\mu}$.  

\subsection*{Algebras of functions on finite--dimensional domains}

Let us discuss a few examples of Banach algebras of functions on finite-dimensional domains to further motivate the constructions on infinite-dimensional spaces studied in this paper.

The \textit{Beurling algebra} is an algebra of functions on the circle defined via a submultiplicative weight function $w \colon \mathbb Z \to \mathbb R_+$ in the Fourier domain, $w(m + n) \leq w(m) w(m)$. For $w(0) =1$, the Fourier image of $\ell^1_w(\mathbb Z)$ is a commutative Banach algebra of continuous functions. 
Since the early works of Beurling and Wermer \cites{Beurling38,Wermer54}, these algebras have been studied extensively in the context of harmonic analysis on locally compact abelian groups (LCAs); e.g., \cite{Kaniuth09}.

The \textit{disc algebra} is the space of holomorphic functions on the unit disc $\Delta \subset \mathbb{C}$ with continuous extension to $\Bar{\Delta}$. The disc algebra has structure of a Banach algebra under pointwise multiplication and the supremum norm.

The Beurling and disc algebras both lack an inner product and associated Hilbert space structure. In contrast, \textit{periodic Hilbert spaces} \cites{Babuska68a,Babuska68b} and \textit{harmonic Hilbert spaces} \cite{Delvos97} are spaces of continuous functions on the unit circle and real line, respectively, that are simultaneously reproducing kernel Hilbert spaces (RKHSs) and Banach algebras under pointwise multiplication. Analogs of periodic and harmonic Hilbert spaces have also been studied in more general LCAs \cites{FeichtingerEtAl07,DasGiannakis23,DasEtAl23}. For an LCA $G$, these spaces are typically built as Fourier images of weighted $\ell^2$ spaces on the dual group, $\hat G$, with the weights $w \colon \hat G \to \mathbb R_+$ being subadditive and/or subconvolutive \cites{Feichtinger79,Grochenig07}. For an RKHA on an LCA $G$ given by the Fourier image of $\ell^2_w(\hat G)$, it can be shown \cite{GiannakisMontgomery25} that the weight $w \colon \hat G \to \mathbb R_+$ must be necessarily subconvolutive, $w^{-2} * w^{-2} \leq C w^{-2}$ 
a.e.\ for a constant $C > 0$. This subconvolutive condition is equivalent to the condition $\Lambda^2 \leq C\Lambda,\; \Lambda \in \ppCone^W_{ess}$, used in Theorem~\ref{thm:main}.

\subsection*{Multi-index notation}
We use capital letters $I$ and $J$ for infinite multi-indices of the form $I=(i_1, i_2, \dots),\, J=(j_1, j_2, \dots)$ where $I$ and $J$ have non-negative integer components. We use $\mathcal{I}$ to represent the space of all such multi-indices with only finitely many non-zero components (i.e., $\sum_r i_r< \infty $). The length of $I \in \mathcal I$ is defined by $|I| := \sum_r i_r $ and we also set $I! := i_1 ! i_2 ! \dots $.

\subsection*{Plan of the paper}
The rest of the paper is organized as follows. In Section~\ref{sec:gaussian_measures} we survey the definitions and basic properties of Gaussian measures and associated Cameron--Martin Hilbert spaces. In Section~\ref{sec:segal_bargmann} we describe the construction of the Segal--Bargmann space over an infinite-dimensional Hilbert space and its multiplicative correspondence with a symmetric Fock space associated with the Cameron--Martin Hilbert space. Section~\ref{sec:proof} contains our proof of Theorem~\ref{thm:main}. Section~\ref{sec:further} further explores the integral representation of $\ppAlg_{\ppWFunc}$ as an RKHS and presents a twisted analog of the canonical commutation relations (CCRs) for creation and annihilation operators on $\ppAlg_{\ppWFunc}$. A list of the main symbols used in the paper is included in Table~\ref{tab:symbols}.

\section{Gaussian measures}
\label{sec:gaussian_measures}

In this section, we first review basic concepts in the theory of infinite-dimensional Gaussian measures on $\mathbb{H}$ such as the covariance operator and the Cameron--Martin Hilbert space $\mathbb{H}_{\mu}$. We then study the corresponding space of square-integrable functions and its inner product structure.
In a real setting, many of the statements and definitions of this section can be found in \cite{Bogachev1998}.

\subsection{Definition and examples}

Let \( X \) be a Banach space, and let \( X^* \) denote the space of all continuous linear functionals on \( X \) (we primarily consider the case of a Hilbert space). For a convex set $C$ in $\mathbb C^n$ and a set of linear functionals \( g_1, \dots, g_n \in X^* \), the preimage of  $C$ under the map defined by $(g_1, \dots, g_n)$ is called a cylindrical set in $X$. The cylindrical $\sigma$-algebra \( \Sigma \) is the $\sigma$-algebra generated by all such cylindrical sets. A measure \( \mu \) on \( (X, \Sigma) \) is Gaussian if for every \( g \in X^* \), the pushforward measure \( g_*\mu \) is a one-dimensional Gaussian random variable.

\begin{exm}
    \label{ex:gaussian}Consider $\ell^2(\mathbb{N})$, the space of square summable sequences $x=(x_1, x_2, \dots )$ of complex numbers with $\sum_j |x_j|^2 < \infty$. A Gaussian measure $\mu$ on \( \ell^2(\mathbb{N}) \) can be defined by declaring the \( x_j \) to be independent (complex) Gaussian random variables with variance \( \sigma_j^2 \). In other words, for a standard linear functional $g_j: \ell^2(\mathbb{N}) \rightarrow \mathbb{C}$ with $g_j(x)=x_j$, the one-dimensional push-forward measure $(g_j)_*\mu$ is a Gaussian distribution with variance $\sigma_j^2$. In order for the measure $\mu$ to be $\sigma$-additive we need the $\sigma_j$ to decay to zero sufficiently rapidly so that $\sum_j \sigma_j^2 < \infty$ (alternatively, if $\sigma_j$ does not decay to zero rapidly-enough then $\mu(\ell^2(\mathbb{N}))=0$).
\end{exm}

For a measure $\mu$ on a Banach space $X$ with the cylindrical $\sigma$-algebra $\Sigma$ 
, the Fourier transform $\hat \mu \colon X^* \to \mathbb C$ is defined as $\hat{\mu}(g) = \int_X e^{i g(x)} \, d\mu(x)$. When $\mu$ is Gaussian we have

$$
\hat{\mu}(g_\xi) = e^{-\frac{1}{4} \langle \xi, \ppfxCov \xi\rangle},
$$
where $g_\xi = \langle \xi, \cdot \rangle$
and $\ppfxCov$ is a strictly positive-definite trace-class operator. Let $\{\xi_1, \xi_2, \dots\} $ be an orthonormal basis of $\mathbb H$ that diagonalizes $\ppfxCov$ with eigenvalues $k_j^2>0$; i.e., $\ppfxCov\xi_j = k_j^2 \xi_j$.  Heuristically, we write \
$$
d\mu = \prod_{j=1}^{\infty} \frac{1}{\pi k_j^2} e^{-|z_j|^2/k_j^2} dz_j \propto e^{-\langle \xi, \ppfxCov^{-1} \xi \rangle},
$$
where $z_j(\xi) = \langle \xi_j, \xi\rangle$ is the coordinate system associated with our fixed basis. Compared to Example~\ref{ex:gaussian}, the random variables $g_j$ are eigenvectors of the covariance operator and the eigenvalues $k_j^2$ are equal to the corresponding variances $\sigma_j$.

As in finite dimensions, Gaussian measures in infinite dimensions are completely characterized by their mean and covariance operators. Here, we always consider centered Gaussian measures where the mean is zero. 

\subsection{Cameron--Martin Hilbert space}
Using the covariance operator, $\ppfxCov$, one canonically associates a Hilbert space $ \mathbb{H}_{\mu} $ to every Gaussian measure $\mu$ on $\mathbb{H}$,
 $$
 \mathbb{H}_{\mu} = \{ \xi \in \mathbb{H}: \langle \xi,\ppfxCov^{-1} \xi \rangle < \infty  \} = \{ \xi \in \mathbb{H}: \|\ppfxCov^{-1/2}\xi\|<\infty \}.
 $$
 This subspace has $\mu$-measure zero and is known as the Cameron--Martin Hilbert space (or RKHS in non-probabilistic contexts) 
 
 associated with $\mu$. Note that the natural inclusion map $\mathbf{i} \colon \mathbb{H}_{\mu} \rightarrow \mathbb{H}$ has a corresponding dual map $\mathbf{i}^*\colon \mathbb{H}^* \rightarrow \mathbb{H}_{\mu}^*$, where the continuous dual spaces $\mathbb H^*$ and $\mathbb H_\mu^*$ are defined with respect to the norm of $\mathbb H$. If we identify $\mathbb{H}$ and $\mathbb{H}^*$ via $\xi_j \leftrightarrow \langle \xi_j, \cdot \rangle_{\mathbb{H}}$ we then obtain an inclusion of Hilbert spaces $\mathbb{H}_{\mu} \subset \mathbb{H} \subset \mathbb{H}_{\mu}^*$. The basis vector $\xi_j$ belongs to $\mathbb{H}_{\mu}^*$ via this identification (expressed as $\mathbf{i}^*(\mathbf{i}(\xi_j))$), and it is not hard to verify that $\| \xi \|^2_{\mathbb{H}_{\mu}^*} = k_j$ while $\| \xi \|^2_{\mathbb{H}_{\mu}} = k_j^{-1}$.  Later,
 we will see that applying the Fock space construction to $\mathbb{H}_{\mu}^*$ yields the space of square-integrable functions on $\mathbb{H}$ with respect to the Gaussian measure $\mu$ (see Proposition ~\ref{prop:FockL2}).

\subsection{Moments of Gaussian measures}\label{subsec:gaussian}
In this subsection, we describe the space of square-integrable functions, $L^2(\mathbb H, \mu)$, on the infinite-dimensional domain, 
 \( \mathbb{H} \), with respect to the Gaussian measure $\mu$. A function $f\colon \mathbb{H} \rightarrow \mathbb{C}$ belongs to $L^2(\mathbb H, \mu)$ if and only if $\int_{\mathbb{H}} |f(\xi)|^2 \, d\mu(\xi) < \infty$. We provide a concrete representation of this space in terms of standard monomials on $\mathbb{H}$. We need to determine the $L^2$ norms and the covariance structure associated with these monomials. We begin with the one-dimensional case involving a complex Gaussian random variable.

\begin{lem}
 \label{lem:1d_gaussian}
 Let $f_j: \mathbb{C} \rightarrow \mathbb{C}$, $f_j(s)=s^j$ be  standard monomials over $\mathbb{C}$. Then the covariance $\langle s^j, s^{j'} \rangle_{\mathbb{C}} = \frac{1}{\pi k^2} \int_{\mathbb{C}} \overline{f_j(s)} f_{j'}(s)\, e^{-|s|^2/k^2}\, ds$, $k>0$
, relative to the Gaussian measure with variance $k^2$, is zero if $j\neq j'$. Otherwise, if $j=j'$ we have
    $$
    \| s^j \|_{\mathbb{C}}^2 = \left< s^j, s^{j} \right>_{\mathbb{C}} = \frac{1}{\pi k^2} \int_{\mathbb{C}} |s|^{2j} e^{-|s|^2/k^2}\, ds =  j!\, k^{2j}.
    $$

\end{lem}
\begin{proof} First, assuming $j\geq j'$, $\langle s^j, s^{j'} \rangle_{\mathbb{C}}$ is defined by integrating $s^{j-j'}$ with respect to a radially symmetric

    To evaluate $\left< s^j, s^{j} \right>_{\mathbb{C}}$, let $\beta = 1/k^2$.  Writing the moments in polar coordinates we have

    $$
    \frac{1}{\pi k^2} \int_{\mathbb{C}} |s|^{2j} e^{-|s|^2/k^2}\, ds  =\frac{\beta}{\pi} \int_0^{2 \pi} \int_0^{\infty} r^{2 j} e^{-\beta r^2} r \, d r \, d \theta.
    $$
    Now using another change of variables, $s=r^2,$ $ds=2 r \,dr$, the claim of the lemma 
    follows from
\begin{align*}
\int_0^{2 \pi} \int_0^{\infty} r^{2 j} e^{-\beta r^2} r d r d \theta & =\pi \int_0^{\infty} t^j e^{-\beta t} d t \\
& =\pi(-1)^j \frac{\partial^j}{\partial \beta^j} \int_0^{\infty} e^{-\beta t} d t
=\pi(-1)^j \frac{\partial^j}{\partial \beta^j}\left(\frac{1}{\beta}\right) \\
& =\pi(-1)^j \frac{(-1)^j j !}{\beta^{j+1}}
=\frac{\pi j !}{\beta^{j+1}} 
=\pi \;  j!\;  k^{2j+2}. \qedhere
\end{align*}
\end{proof}

Let us now derive covariance formulas in the case of the infinite-dimensional Gaussian measure $\mu$ with covariance operator $\ppfxCov$. Given the orthonormal basis $\xi_1, \xi_2, \dots$ of $\mathbb{H}$ diagonalizing $\ppfxCov$, we consider  $z_j$ as a function on $\mathbb{H}$ defined by $z_j(\xi) = \left<\xi, \xi_j \right>$. More generally, for a multi-index $I=(i_1, i_2, \dots) \in \mathcal{I}$, we have
\begin{equation*}
    z^I(\xi) = \left<\xi, \xi_1\right>^{i_1} \left<\xi, \xi_2\right>^{i_2} \dots, \quad \quad \xi \in \mathbb{H}.
\end{equation*}

A polynomial function on $\mathbb{H}$ has the general form $\sum_{I} c_I z^I$, where the sum is taken over $I \in \mathcal{I}$ and has finitely many nonzero coefficients $c_I$.

\begin{prop}\label{prop:gaussian_covariance}
    Let $\mu$ be a Gaussian measure on $\mathbb{H}$ with covariance operator $\ppfxCov$. Given multi-indices $I,I'$, the covariance expression 
    $$\left< z^I, z^{I'} \right>_{L^2(\mathbb{H}, \mu)} = \int_{\mathbb{H}} \bar z^I(\xi) z^{I'}(\xi) \, d\mu (\xi)  $$
    associated with the measure $\mu$ is zero if $I\neq I'$. Moreover, we have
    $$
    \| z^I \|_{L^2(\mathbb{H}, \mu)} =  \sqrt{I!}\, k_1^{i_1} k_2^{i_2} \dots.
    $$
    \end{prop}
\begin{proof}
    Since $\mu$ is the product of independent Gaussian measures, the above integral 
    
    with respect to $\mu$ splits into one-dimensional integrals of $z_j$'s.  Then, the fact that for $j\neq j'$ in Lemma~\ref{lem:1d_gaussian}, $\left< s^j, s^{j} \right>_{\mathbb{C}} =0$  implies that $\left< z^I, z^{I'} \right>_{L^2(\mathbb{H}, \mu)} = 0 $ for $I\neq I'$. 
    Again, using Lemma~\ref{lem:1d_gaussian}, if $I=I'$,  we can write $\| z^I \|^2_{L^2(\mathbb{H}, \mu)}$ as
    $$
    \int_{\mathbb{C}} |z|^{2I} \, d\mu = \prod_{j} \frac{1}{\pi k_j^2} \int_{\mathbb{C}} |z_j|^{2 i_j} e^{-|z_j|^2/k_j^2}\, dz_j = \prod_{j} \left( i_j!\; k_j^{2i_j} \right) =  I!\, k^{2I},
    $$
    where $k_j$ are the eigenvalues of $\ppfxCov$, and $k^{2I} = k_1^{2i_1} k_2^{2i_2}\dots$. 
\end{proof}

\section{Infinite-dimensional Segal--Bargmann space}
\label{sec:segal_bargmann}

\subsection{Holomorphic functions on a Banach space}
\label{sec:holomorphic}

In this section, we provide the definition and preliminary results concerning holomorphic functions on an infinite-dimensional Banach space $X$ (in this paper $X$ is assumed to be the Hilbert space $\mathbb{H}$). Most of the definitions and statements provided in this section can be found in \cite{mujica2010}.  

\newcommand{\ML}[0]{\mathcal{L}}
\newcommand{\BML}[0]{\mathcal{BL}}
\newcommand{\Pol}[0]{\mathcal{P}}
\newcommand{\BPol}[0]{\mathcal{BP}}

\newcommand{\BMLnorm}[1]{\| #1 \|_{\BML}}
\newcommand{\Polnorm}[1]{\| #1 \|_{\BPol}}

We denote by $\ML(X^m)$ the space of all $m$-linear maps from $m$-tuples in $X^m$ to $\mathbb{C}$. We define $\BML(X^m) \subseteq \mathcal L(X^m)$ as the Banach space of all $m$-linear maps which are continuous with respect to the norm
$$
\BMLnorm{q} = \sup_{\|x_j\|<1} |q(x_1, \dots, x_m)|;
$$
 that is, $q \in \BML(X^m)$ iff $\BMLnorm{q} < \infty$. A symmetric $m$-linear form, $q$, satisfies $q(x_{\sigma(1)}, \dots, x_{\sigma(m)}) = q(x_1, \dots, x_m)$ for any permutation $\sigma \in S_m$. The subspace of symmetric maps in $\BML(X^m)$ is denoted by $\BML^s(X^m)$. We use the shorthand notation $q(x_1^i x_2^j)$ for $q(x_1, \dots, x_1, x_2, \dots, x_2)$ where $x_1$ and $x_2$ are repeated $i$ and $j$ times, respectively, and $i+j=m$ 

\begin{defn}
    A map $p \colon X\rightarrow \mathbb{C}$ is an $m$-homogeneous polynomial if $p(x)=q(x^m)=q(x,\dots, x)$ for some $m$-linear map $q \in \ML(X^m)$. 
\end{defn}

We denote the space of $m$-homogeneous polynomials by $\mathcal P_m(X)$. Similarly to $\mathcal{BL}(X^m)$, we define the Banach space $\mathcal B\mathcal P_m(X) \subseteq \mathcal P_m(X)$ of $m$-homogeneous polynomials using the norm $\Polnorm{p} = \sup_{\| x\| \leq 1} |p(x)|$. The map that sends $q \in \BML(X^m)$ to the corresponding $p \in \Pol_m(X)$ induces a topological isomorphism from $\BML^s(X^m)$ to $\mathcal{BP}_m(X)$ \cite[Corollary~2.3]{mujica2010}. Furthermore, one can prove that $\Polnorm{p} \leq \BMLnorm{q} \leq \frac{m^m}{m!} \Polnorm{p}$.  

A polynomial of degree at most $m$ is a function of the form $p(x) = p_1(x) + \dots + p_m(x)$ where $p_j \in \Pol_j (X)$. We can now define power series as follows.

\begin{defn}
\label{def:power_series}
    Let $a\in X$. A power series  centered at $a$ is a series of mappings from $X$ to $\mathbb{C}$ 
    of the form $x \mapsto \sum_{m=0}^{\infty} p_m(x-a)$, $x \in X$, such that $p_m \in \BPol_m (X)$. The radius of convergence, $R$, is defined by $$R = \sup \left\{r\geq 0: \text{$\sum_{m=0}^{\infty} p_m(x-a)$ converges uniformly on $\overline{B(a,r)}$} \right\}. $$ 
\end{defn}

\begin{defn}
\label{def:holomorphic}
    Let $U$ be an open subset of $X$. A function $f \colon U \rightarrow \mathbb{C}$ is holomorphic if for every $a \in U$ there exist $r>0$ and $p_m \in \BPol_m(X)$ such that the power series $f(x) = \sum_{m=0}^{\infty} p_m(x-a)$ converges uniformly on $B(a,r)$. We write $\Hol(U)$  for the space of all holomorphic functions on $U$.  
    
\end{defn}

We now prove that, just as in the finite dimensions, a power series centered at the origin defines a holomorphic function on $U=B(0,R)$, where $R$ is the radius of convergence from Definition~\ref{def:power_series}.

\begin{prop}\label{prop:hol-origin}
    Let $f(x) = \sum_{m=0}^{\infty} p_m(x)$ be a power series centered at the origin with radius of convergence $R>0$. Then $f \in \Hol (U)$ where $U=B(0,R)$. 
\end{prop}
\begin{proof}
    Let $a\in B(0,R)$. To prove the statement we have to find polynomials $\tilde{p}_m \in \BPol_m(X)$ such that $f(x) = \sum_{m=0}^{\infty} \tilde{p}_m(x-a)$ converges uniformly to $f(x)$ on some $B(a,r) \subset B(0,R)$. For $q_m \in \ML(X^m) $ and $k\leq m$, define $p_{m,k} \in \Pol_k(X)$ via
    $$
    q_{m,k}(x_1, \dots, x_k) = {m \choose k} q_m(a,\dots, a, x_1,\dots, x_k),\quad  p_{m,k}(x) = q_{m,k}(x^k).
    $$
    Heuristically, we can write 
    \begin{align}
        \nonumber \sum_{m=0}^{\infty} p_m(x) &= \sum_{m=0}^{\infty} p_m(x-a+a)\\
        \nonumber &= \sum_{m=0}^{\infty} q_m\left( (x-a+a)^{m} \right) \\
        \nonumber &= \sum_{m=0}^{\infty} \sum_{k=0}^m {m \choose k} q_m(a^{m-k} (x-a)^k)\\
        &= \sum_{m=0}^{\infty} \sum_{k=0}^m  p_{m,k}(x-a) \label{eq:poly_expansion_beforelast}   \\ 
        &=  \sum_{k=0}^{\infty} \sum_{m=k}^{\infty} p_{m,k}(x-a),   \label{eq:poly_expansion_last}     
    \end{align}
    suggesting that we should take $\tilde{p}_k(x):= \sum_{m=k}^{\infty} p_{m,k}(x)$ so that $$f(x)=\sum_{m=0}^{\infty} \tilde{p}_m(x-a).$$
    Of course, we need to prove that $\tilde{p}_m \in \BPol_m(X)$ and that the above sum is uniformly convergent. The latter, also implies that we can reorder the summation to get (\ref{eq:poly_expansion_last}) from (\ref{eq:poly_expansion_beforelast}). Observe that if $\|x-a\| < r$ then 
    \begin{align}
        \nonumber \left\lVert\sum_{m=0}^\infty\sum_{k=0}^\infty {m \choose k} q_m(a^{m-k} (x-a)^k) \right\rVert &  \leq \sum_{m=0}^{\infty} \sum_{k=0}^m {m \choose k} \BMLnorm{q_m} \|a\|^{m-k} \|x-a\|^k \\
         \label{eq:expand_norm}& \leq  \sum_{m=0}^{\infty}  \BMLnorm{q_m}\; \left(\|a\| +r \right)^m.
    \end{align}
    
    This implies that if the sum on the right-hand side of~\eqref{eq:expand_norm} is convergent then we can interchange the summation over $m$ and $k$ in \eqref{eq:poly_expansion_beforelast} to arrive at~\eqref{eq:poly_expansion_last}. 
To verify that the sum in~\eqref{eq:expand_norm} is convergent, let $R$ be the radius of convergence of $\sum_{m=0}^{\infty} p_m(x)$. Then, by the Cauchy-Hadamard formula \cite[Theorem~4.4]{mujica2010} we have
$$R^{-1} = \limsup_{m\rightarrow \infty} \Polnorm{p_m}^{1/m} < \infty.$$
Now if $B(a,r) \subset B(0,R)$ we have 
\begin{align*}
\|a\|+r < R &= \left( \limsup_{m\rightarrow \infty} \Polnorm{p_m}^{1/m} \right)^{-1}\\
& \leq \left( \limsup_{m\rightarrow \infty} \frac{(m!)^{1/m}}{m} \BMLnorm{q_m}^{1/m} \right)^{-1}\\
& \leq  \left( \limsup_{m\rightarrow \infty}  \BMLnorm{q_m}^{1/m} \right)^{-1}, 
\end{align*}
where in the last inequality we used $\lim_{m\rightarrow \infty} \frac{(m!)^{1/m}}{m} = \frac{1}{e}$. This implies that the right hand side of (\ref{eq:expand_norm}) is convergent. From (\ref{eq:expand_norm}), we also have that 
$$ \sum_{m=k}^{\infty} {m \choose k} \BMLnorm{q_m} \|a\|^{m-k} r^k < \infty,
$$
hence $\tilde{p}_k \in \BPol_k(X)$, which finishes the proof.
\end{proof}

\subsection{Segal--Bargmann space}
The subset of holomorphic functions of the form $f(\xi)= \sum_{I \in \mathcal I} a_I z^I(\xi)$ in $L^2(\mathbb{H}, \mu)$ is denoted by $\mathcal{H}L^2(\mathbb{H}, \mu)$. This is the infinite-dimensional Segal--Bargmann space associated with the measure $\mu$. Note that for every such function $f$ we have $\sum_{I \in \mathcal I} |a_I|^2 \| z^I\|^2_{L^2(\mathbb{H}, \mu)} < \infty $. For a treatment of finite- and infinite- dimensional Segal--Bargmann spaces and their relevance in mathematical physics see~\cite{MR1770752}.

\subsection{Fock space}\label{subsec:fock}
Let $\ppHil$ be a Hilbert space, and let $\ppHil^{\otimes n}$ be the n-fold tensor product of $\ppHil$. The subspace of symmetric tensors in $\ppHil^{\otimes n}$ is denoted by $\ppHil^{\vee n}$. We define the symmetric Fock space $F(\ppHil)$ associated with $\ppHil$ as the closure of the symmetric tensor algebra $ T^\vee(H) := \bigoplus_{n=0}^{\infty} \ppHil^{\vee n}$ with respect to the norm induced by the inner product $\langle \cdot , \cdot \rangle_{F(\ppHil)} = \sum_{n=0}^{\infty} n! \, \langle \cdot , \cdot \rangle_{\ppHil^{\vee n}}$ .

\begin{rk}
    The symmetric Fock space is typically constructed (for instance, see \cite{GiannakisEtAl25}) by taking the closure of $T^\vee(\mathbb H)$ with respect to an inner product, $\sum_{n=0}^{\infty} \langle \cdot , \cdot \rangle_{\ppHil^{\vee n}}$, that does not include $n!$ weights. For our purposes, the presence of $n!$ will be crucial for establishing a correspondence between a space of functions on $\mathbb H$ and a symmetric Fock space that preserves both multiplication and inner product structures simultaneously (multiplication on the Fock space is given by $\vee$, discussed below).
\end{rk}

\noindent Let $\eta_1, \dots, \eta_n$ be vectors in $\ppHil$, and define
\begin{align}\label{eq:def_wedge}
 \eta_1 \vee \dots \vee \eta_n = P_{+}(\eta_1 \otimes \dots \otimes \eta_n) := \frac{1}{n!} \sum_{\sigma \in S_n} \eta_{\sigma(1)} \otimes \dots \otimes \eta_{\sigma(n)},
\end{align}
where $P_+$ is the orthogonal projection from $\ppHil^{\otimes n}$ onto $\ppHil^{\vee n}$ and the sum is taken over all permutations in $S_n$. The symbol $\vee$ defines a multiplication structure on the Fock space. Let $I=(i_1, i_2, \dots)$ be a multi-index in $\mathcal{I}$. For vectors  $ \xi_1, \xi_2, \dots  $ in $\ppHil$, we use the shorthand notation $\xi^{\vee I}$ to denote  $\xi_1^{i_1} \vee \xi_2^{i_2} \vee \dots$ where $\xi_1^{i_1}$ represents $\xi_1 \vee \dots \vee \xi_1$ (repeated $i_1$ times).

\begin{lem}\label{lem:fock_mult}
    Given $\{ \xi_1, \xi_2, \dots \} \subset H$, multi-indices $I=(i_1, i_2, \dots )$ and $J=(j_1, j_2, \dots )$ in $\mathcal I$, we have $\xi^{\vee I} \vee \xi^{\vee I} = \xi^{\vee (I+ J) }$, i.e.,  $$\left( \xi_1^{i_1} \vee \xi_2^{i_2} \vee \dots \right) \vee \left( \xi_1^{j_1} \vee \xi_2^{j_2} \vee \dots \right) = \xi_1^{i_1 + j_1} \vee \xi_2^{i_2+j_2} \vee \dots. $$
    
\end{lem}
\begin{proof}
    For $n=|I|$ and $m=|J|$, we define subscripts $k_1, \dots, k_n$ and $\ell_1, \dots, \ell_m$ such that
    \begin{align*}
    \xi^{\vee I} &= P_{+}(\xi_{k_1}\otimes \dots \otimes \xi_{k_n}),\\
    \xi^{\vee J} &= P_{+}(\xi_{\ell_1}\otimes \dots \otimes \xi_{\ell_n}).
    \end{align*}
    In other words, $i_s = \#\{r: k_r = s \} $. Then,
    \begin{align*}
         \xi^{\vee I} \vee \xi^{\vee J} &= P_+\left( P_+(\xi_{k_1}\otimes \dots \otimes \xi_{k_n}) \otimes P_+(\xi_{\ell_1}\otimes \dots \otimes \xi_{\ell_n}) \right) \\
         &= P_+\left(\left(\frac{1}{n!} \sum_{\sigma \in S_n} \xi_{k_{\sigma(1)}}\otimes \dots \otimes \xi_{k_{\sigma(n)}}\right)\otimes \left(\frac{1}{m!} \sum_{\sigma' \in S_m} \xi_{\ell_{\sigma'(1)}}\otimes \dots \otimes \xi_{\ell_{\sigma'(m)}}\right)\right) \\
         &=\frac{1}{n! m!} \sum_{\sigma \in S_n}  \sum_{\sigma' \in S_m}P_+\left(\xi_{k_{\sigma(1)}}\otimes \dots \otimes \xi_{k_{\sigma(n)}}\otimes  \xi_{\ell_{\sigma'(1)}}\otimes \dots \otimes \xi_{\ell_{\sigma'(m)}}\right)\\
         &=\frac{1}{n! m!} \sum_{\sigma \in S_n}  \sum_{\sigma' \in S_m}P_+\left(\xi_{k_1}\otimes \dots \otimes \xi_{k_n}\otimes  \xi_{\ell_1}\otimes \dots \otimes \xi_{\ell_m} \right)\\
         &=P_+\left(\xi_{k_1}\otimes \dots \otimes \xi_{k_n}\otimes  \xi_{\ell_1}\otimes \dots \otimes \xi_{\ell_m} \right)\\
         &= \xi^{\vee (I+J)}. \qedhere
    \end{align*}

\end{proof}

\begin{lem}\label{lem:tensor_norm}
    Let $\{ \xi_1, \xi_2, \dots \}$ be an orthogonal basis of $H$. Then, $\| \xi^{\vee I}\|_{F(\ppHil)} = \sqrt{I!}\, \| \xi_1 \|^{i_1} \| \xi_2 \|^{i_2} \dots $. 
\end{lem}
\begin{proof}
    Let $n=|I|$. By definition,
    $$
    \| \xi^{\vee I}\|^2_{F(\ppHil)} = \| \xi_{k_1} \vee \dots \vee \xi_{k_n}\|^2_{F(\ppHil)} = n! \left\| \frac{1}{n!} \sum_{\sigma \in S_n} \xi_{k_{\sigma(1)}} \otimes \dots \otimes \xi_{k_{\sigma(n)}} \right\|^2,
    $$
    and the right-hand side can be written as
    $$
   \frac{1}{n!} \sum_{\sigma \in S_n} \left<  \xi_{k_{\sigma(1)}} \otimes \dots \otimes \xi_{k_{\sigma(n)}}, \sum_{\sigma' \in S_n} \xi_{k_{\sigma'(1)}} \otimes \dots \otimes \xi_{k_{\sigma'(n)}} \right>.
    $$
    Expanding this sum, we get a contribution of $\| \xi_1 \|^{2i_1} \| \xi_2 \|^{2i_2} \dots $ each time we have an exact match of indices. The list of subscripts $k_1, k_2, \dots$ contains $i_1$ copies of $1$, $i_2$ copies of $2$, etc. As a result, each of $n!$ index permutations gives rise to $I!= i_1!\, i_2! \, \dots$ terms with the total of $n! I!$ contributions of $\| \xi_1 \|^{2i_1} \| \xi_2 \|^{2i_2} \dots $. The claim of the lemma follows by applying the multiplicative factor $\frac{1}{n!}$ and taking the square root. 
\end{proof}

 We define the weighted Fock space $F_w(H) $ as the closure of the symmetric tensor algebra $ T^\vee(H)$ with respect to the norm induced by the inner product $\langle \cdot , \cdot \rangle_{F_w(\ppHil)} = \sum_{n=0}^{\infty} w(n)^2 n! \, \langle \cdot , \cdot \rangle_{\ppHil^{\vee n}}$. More specifically, in terms of the fixed orthonormal basis of Lemma \ref{lem:tensor_norm}, the inner product of $F_w(H)$ is defined by extending
$$
\langle \xi^{\vee I}, \xi^{\vee J} \rangle_{F_w(H)} := \delta_{IJ} w(|I|)^2 I!.
$$
As before, vectors $\xi^{\vee I}, \, I \in \mathcal{I}$, form an orthogonal basis of $F_w(H)$ with $
\| \xi^{\vee I} \|_{F_w(H)} = w(I)\; \| \xi^{\vee I} \|_{F(H)}$.

\subsection{Multiplicative correspondence to Fock space}

\catcode`\@=11
\newdimen\cdsep
\cdsep=3em

\def\cdstrut{\vrule height .6\cdsep width 0pt depth .4\cdsep}
\def\@cdstrut{{\advance\cdsep by 2em\cdstrut}}

\def\arrow#1#2{
  \ifx d#1
    \llap{$\scriptstyle#2$}\left\downarrow\cdstrut\right.\@cdstrut\fi
  \ifx u#1
    \llap{$\scriptstyle#2$}\left\uparrow\cdstrut\right.\@cdstrut\fi
  \ifx r#1
    \mathop{\hbox to \cdsep{\rightarrowfill}}\limits^{#2}\fi
  \ifx l#1
    \mathop{\hbox to \cdsep{\leftarrowfill}}\limits^{#2}\fi
}
\catcode`\@=12
\cdsep=3em
The Segal--It\^o type isomorphisms provide an identification of a Fock space with an appropriate space of square integrable functions with respect to a Gaussian measure. In the holomorphic case this isomorphism takes a simpler form. The map  $\theta\colon F(\mathbb{H}^*_{\mu}) \rightarrow L^2(\mathbb{H}, \mu)$ defined by extending $\theta (\xi^{\vee I}) = z^I$ is not just an isometry, but, in fact, it identifies the multiplication structure of $F(\mathbb{H}_{\mu}^*)$ and $L^2(\mathbb{H}, \mu)$. The choice of $n!$ in the weighting of the Fock space, and the $1/(n!)$ factor in the definition of $\vee$, (see (\ref{eq:def_wedge})), are both essential to preserve norm and multiplicative structures simultaneously.

The following proposition does not appear to have been stated and proved in this form, although it is most likely known implicitly in the literature (see \cite{Segal1963}, and \cite[Proposition~2.1]{Trieu2017} for a similar observation).

\begin{prop}\label{prop:FockL2}
    Let $m$ be the point-wise multiplication of functions in $\mathcal{H}L^2(\mathbb{H}, \mu)$, and let $\vee$ be the tensor product multiplication in the Fock space $F(\mathbb{H}^*_{\mu})$. Let $\mathcal{D}_{m}$ be the subspace polynomials in $\mathcal{H}L^2(\mathbb{H}, \mu)$, and let  $\mathcal{D}_{\vee}=T(\mathbb{H}^*_{\mu})$ be the tensor algebra in $F(\mathbb{H}^*_{\mu})$. Then, $\theta (\xi^{\vee I}) = z^I$ extends to an isometric isomorphism $\theta\colon F(\mathbb{H}^*_{\mu}) \rightarrow \mathcal{H}L^2(\mathbb{H}, \mu)$ 
    such that the following diagram commutes,

$$
\begin{matrix}

  \mathcal{D}_{\vee} \subset F(\mathbb{H}^*_{\mu}) \times F(\mathbb{H}^*_{\mu})                    & \arrow{r}{\vee}   & F(\mathbb{H}^*_{\mu})                    \cr
  \arrow{d}{\theta \times \theta} &                      & \arrow{d}{\theta} \cr
  \mathcal{D}_m \subset \mathcal{H}L^2(\mathbb{H}, \mu)\times \mathcal{H}L^2(\mathbb{H}, \mu)                  & \arrow{r}{m} & \mathcal{H}L^2(\mathbb{H}, \mu)                  \cr
\end{matrix}
$$
\end{prop}
\begin{proof}
    Let $\xi_1, \xi_2, \dots \in \mathbb{H}$ be an orthonormal eigenbasis of $\ppfxCov$. Then in terms of the Cameron--Martin Hilbert space $\mathbb{H}_\mu$ we have $\| \xi_j \|_{\mathbb{H}_{\mu}} = 1/k_j$ which implies that $\| \xi_j \|_{\mathbb{H}_{\mu}^*} = k_j$. Note that from Lemma~\ref{lem:tensor_norm} we have  $\| \xi ^{\vee I} \|_{F(\mathbb{H}_{\mu})} = \sqrt{I!} \| \xi_1 \|^{i_1} \| \xi_2 \|^{i_2} \dots$, but $$\| \xi ^{\vee I} \|_{F(\mathbb{H}_{\mu}^*)} = \sqrt{I!} \| \xi_1 \|^{i_1}_{\mathbb{H}_{\mu}^*} \| \xi_2 \|^{i_2}_{\mathbb{H}_{\mu}^*} \dots = \sqrt{I!} k_1^{i_1} k_2^{i_2} \dots$$
    Using Proposition~\ref{prop:gaussian_covariance}, this implies that the map from $F(\mathbb{H}_{\mu}^*)$ to $\mathcal{H}L^2(\mathbb{H}, \mu)$ defined by $\theta(\xi^{\vee I}) = z^I$ is an isometry (since it maps an orthogonal basis to an orthogonal basis and preserves the norms of the basis elements).  From Lemma~\ref{lem:fock_mult} it follows that the multiplication structure given by $\xi^{\vee I} \vee \xi^{\vee J} = \xi^{\vee (I+J)}$ is consistent with the point-wise multiplication, $z^I z^J = z^{I+J}$, in $\mathcal{H}L^2(\mathbb{H}, \mu)$. Fix $(\sum_{I \in \mathcal I} a_I \xi^{\vee I}, \sum_{I \in \mathcal I} b_J \xi^{\vee J})$ in $\mathcal{D}_{\vee}$.  Then  $\theta(\sum_{I \in \mathcal I} a_I \xi^{\vee I}) = \sum_{I \in \mathcal I} a_I z^I \in \mathcal{H}L^2(\mathbb{H}, \mu)$ and since
$$
m\left(\theta\left(\sum_{I \in \mathcal I} a_I \xi^{\vee I}\right),  \theta \left(\sum_{I \in \mathcal I} b_J \xi^{\vee J}\right)\right) = \theta \left(\sum_{I \in \mathcal I} a_I \xi^{\vee I} \vee \sum_{I \in \mathcal I} b_J \xi^{\vee J}\right)
$$
we conclude that the above diagram commutes. 
\end{proof}

\section{Proof of main result}
\label{sec:proof}

 In this section, we provide a proof of Theorem ~\ref{thm:main}. Throughout this section, let $\ppWFunc \in \ppCone^W$ and  $w:\mathbb{N} \rightarrow \mathbb{ R}_+$ with $\lambda_n^{-1} = n! w(n)^2 $. In addition, we use the abbreviated notation $\mathcal K \equiv \mathcal K_\Lambda$ for the reproducing kernel of $\mathcal A_\Lambda$.
 We organize the proof of Theorem~\ref{thm:main} in three parts:
\begin{enumerate}[(i), wide]
\item Let $\mathcal{Q}_w$ be the Hilbert space completion of polynomials on $\bar{\ppBall} \subset \ppfxHil_{\mu}$ with respect to the inner product $\langle z^I, z^J \rangle_w = \delta_{IJ} w(|I|)^2 I! k^{2I}$. We show that $\mathcal{Q}_w$ is an RKHS that can be characterized as 
  $$
 \mathcal{Q}_w = \left\{f: \ppBall \rightarrow \mathbb{C}: f(\xi) = \sum_{I \in \mathcal I} a_I z^I(\xi), \;  \sum_{I \in \mathcal I} |a_I|^2 w(|I|)^2 I! k^{2I} < \infty \right\}.
 $$
Moreover, the restrictions of elements of $\mathcal Q_w$ on $\mathbb B$ are holomorphic functions in the sense of Definition~\ref{def:holomorphic}. 
\item Let $\ppAlg_{\ppWFunc}$ be the RKHS on $\bar \ppBall$ associated with the kernel function $\ppKer(\xi,\eta) = \ppWFunc (\langle \xi, \eta \rangle_{\ppfxHil_{\mu}})$. 
We show that $\ppAlg_{\ppWFunc}$ coincides with $\mathcal{Q}_w$. 
\item We show that if $\ppWFunc \in \ppCone^W_{ess}$, $\ppAlg_{\ppWFunc}$ is a unital RKHA. In particular, $\ppAlg_{\ppWFunc}$ is closed under pointwise multiplication of functions over $\bar{\ppBall}$. Our proof is based on a comparison with the spaces $\mathbf F_w(H)$ studied in \cite{GiannakisEtAl25}, for the choice $H = \mathbb H_\mu^*$. 
\end{enumerate}

\subsection*{Part (i). RKHS on $\bar{\ppBall}$}
First, we show that any formal power series
$ f \in \mathcal{Q}_w $ defines a holomorphic function on $\ppBall$.
Next, we prove that pointwise evaluation is a continuous linear functional
on $\bar\ppBall$. Finally, we show that the reproducing kernel Hilbert space
generated by these evaluation functionals is indeed $\mathcal{Q}_w$.

\subsubsection*{Formal power series $f\in \mathcal{Q}_w$ is a holomorphic function on $\ppBall$} We follow Definition~\ref{formula:Acoeff}. In order to show that $f \in \mathcal Q_w$ is holomorphic, we use Proposition~\ref{prop:hol-origin} with $\ppBall = B(0,1)$ as the unit ball in $\ppfxHil_{\mu}$. In the case of our power series $f(\xi)= \sum_{I \in \mathcal I} a_I z^I(\xi)$ we have $f(\xi) = \sum_{n=0}^\infty p_n(\xi)$ with $p_n(\xi) = \sum_{I: |I|=n} a_I z^I(\xi)$.

First, we need to show that each $p_n$ is continuous, i.e. $p_n \in \BPol_n (\ppfxHil_{\mu})$. Note that $p_n$ is a finite linear combination of monomials $z^I$, and continuity of $p_n$ follows from continuity of each $z^I$. Also, recall that a polynomial $p: \ppfxHil_{\mu} \rightarrow \mathbb{C}$ is continuous if the set $\{ |p(\xi)|: \|\xi\|_{\ppfxHil_{\mu}} \leq 1 \}$ is bounded. To see that $z^I$ is continuous, we write $z^I = z_{j_1} \dots  z_{j_n}$ where $n=|I|$. Then, using Cauchy--Schwarz inequality we get 
$$|z^I(\xi)| = |\langle \xi, \xi_{j_1} \rangle_{\ppfxHil} | \dots  |\langle \xi, \xi_{j_1} \rangle_{\ppfxHil} | \leq \| \xi \|_{\mathbb{H}}^n \leq \| \xi \|_{\ppfxHil_{\mu}}^n \leq 1, $$
provided $\| \xi \|_{\ppfxHil_{\mu}} \leq 1$. This shows that $p_n$ is bounded. 

Second, we need to prove that the sum $f(\xi)= \sum_{I \in \mathcal I} a_I z^I(\xi)$ is uniformly convergent on $\ppBall \subset \ppfxHil_{\mu}$. To show uniform convergence, we let $\mathcal{F}$ denote the set of all multi-indices in a tail of this sum. By the Cauchy--Schwarz inequality,
\begin{equation}
 \label{eq:tail_ineq}
    \left|\sum_{I\in \mathcal{F}} a_I z^I(\xi) \right| = \left|\sum_{I\in \mathcal{F}} (a_I \| z^I\|_w) \left(\frac{z^I(\xi)}{\| z^I\|_w}\right) \right|
    \leq \left( \sum_{I\in \mathcal{F}} |a_I|^2 \| z^I\|^2_w) \right)^{1/2} \left( \sum_{I\in \mathcal{F}}  \frac{|z^I(\xi)|^2}{\| z^I\|^2_w}\right)^{1/2}.
\end{equation}
Note that $\| f\|^2_w = \sum_{I} |a_I|^2 \| z^I\|^2_w < \infty$. Define 
\begin{equation*}
    A(r) = \sup_{\| \xi\|_{\ppfxHil_{\mu}}\leq r} \sum_{I \in \mathcal I} \frac{|z^I(\xi)|^2}{\| z^I\|^2_w},
\end{equation*}
which is independent of $\xi$. It is not hard to see that if $A(r)<\infty$ for every $0<r<1$, then the right-hand side of~\eqref{eq:tail_ineq} tends to zero as the tail becomes smaller, and the series converges uniformly. We have
\begin{equation*}
    A(r) = \sup_{\| \xi\|_{\ppfxHil_{\mu}}\leq r} \sum_{I \in \mathcal I} \frac{|z^I(\xi)|^2}{w(|I|)^2 I! k^{2I}} 
    = \sup_{\| \xi\|_{\ppfxHil_{\mu}}\leq r} \sum_{n=0}^{\infty} \frac{1}{w(n)^2 n!} \left( \| \xi \|^{2n}_{\ppfxHil_{\mu}} \right)
    = \sup_{\| \xi\|_{\ppfxHil_{\mu}}\leq r} \Lambda(\| \xi\|_{\ppfxHil_{\mu}} ),
\end{equation*}
which is finite given our choice $\Lambda \in \ppCone^{W}$. This proves that $ f\in \mathcal{Q}_w$ is holomorphic on $\ppBall$. More generally, it is not hard to see that using a similar argument the expansion $\ppKer(\xi,\eta) = \sum_{n=0}^{\infty} \lambda_n \langle \xi, \ppfxCov^{-1} \eta \rangle^n$ is uniformly convergent for $\xi,\eta \in \bar\ppBall$. \qed

\subsubsection*{For $\xi \in \bar\ppBall$, the pointwise evaluation $\chi_{\xi}: \mathcal{Q}_{w} \rightarrow \mathbb{C}$, defined by $\chi_{\xi}(f) = f(\xi)$, is a continuous linear functional.}
Define $\alpha_j = \langle \xi_j, \xi \rangle_{\ppfxHil}$ and $\alpha = (\alpha_1, \alpha_2, \dots)$ so that $\xi = \sum_j \alpha_j \xi_j$. From definition of $z_j$ we have $\chi_{\xi}(z_j) = \alpha_j $. For a general monomial $z^I$ we get $\chi_{\xi}(z^I) = \alpha^I$. Monomials $z^I, I\in \mathcal{I}$ form an orthogonal basis for $\mathcal{Q}_{w}$ and hence a function $f: \mathbb{H} \rightarrow \mathbb{C}$ belongs to $\mathcal{Q}_{w}$ if and only if
\begin{equation}\label{formula:Acoeff}
f = \sum_{I} a_I z^I \textit{ with } \sum_{I \in \mathcal I} |a_I|^2 \| z^I \|^2_{w} < \infty
\end{equation}
where $\|z^I\|^2_w = w(|I|)^2 I! k^{2I}$. The evaluation map $\chi_{\xi}$ defines a bounded linear functional on $\mathcal{Q}_{w}$ iff there exists $\ppKer_{\xi} = \sum_{I} b_I z^I \in \mathcal{Q}_{w}$ such that $\chi_{\xi}(\cdot) = \left< \ppKer_{\xi}, \cdot \right>_{w}$. In that case,
$$
\alpha^I = \chi_{\xi}(z^I) = \left< \sum_{I} b_J z^J, z^I \right>_{w} =\bar b_I \| z^I \|_{w}^2,
$$
giving
$$
b_I = \frac{\bar z^I(\xi)}{\| z^I \|_{w}^2}, \quad \ppKer_{\xi} = \sum_{I \in \mathcal I} \frac{\bar z^I(\xi)}{\| z^I \|_{w}^2} z^I.
$$
Thus, by~\eqref{formula:Acoeff}, $\ppKer_{\xi}$ lies in $\mathcal Q_w$ and $\chi_\xi$ is continuous iff
\begin{equation}
\label{eq:k_xi_bound}
\| \ppKer_{\xi} \|_{w}^2 = \sum_{I \in \mathcal I} \frac{|\bar z^I(\xi)|^2}{\| z^I \|^4_{w} } \| z^I \|^2_{w} = \sum_{I \in \mathcal I} \frac{| z^I(\xi)|^2}{\| z^I \|^2_{w} } < \infty.
\end{equation}
To verify this, observe that $\| z^I \|^2_{w} = I! \,  w(|I|)^2 k_1^{i_1} k_2^{i_2}\dots$, and hence
$$
\| \ppKer_{\xi} \|_{w}^2 = \sum_n \sum_{|I| = n} \frac{| z^I(\xi)|^2}{I! \, w(|I|)^2 k^{2I}} = \sum_n \frac{1}{n! w(n)^2} \sum_{|I|=n} \frac{n!}{I!} |z^I(\ppfxCov^{-1}\xi)|^{2}.
$$
Using the binomial theorem, the last sum gives $\| \ppfxCov^{-1/2}\xi \|^{2n}_{\mathbb{H}}$. As a result \eqref{eq:k_xi_bound} holds and $\chi_{\xi}$ is a character on $\mathcal{Q}_w$ if and only if
\begin{equation*}
     \sum_n \frac{\| \ppfxCov^{-1/2}\xi \|^{2n}_{\mathbb{H}}}{n! \, w(n)^2}< \infty.
\end{equation*}
Evaluating $\ppWFunc$ at $1\in \mathbb{C}$, we have $\ppWFunc(1) =
\sum_n \frac{1}{n! \, w(n)^2} < \infty$. Since $\| \ppfxCov^{-1/2}\xi \|^{2}_{\mathbb{H}} \leq 1$ for any $\xi \in \ppBall$ the above sum is convergent and $\ppKer_{\xi} \in \mathcal{Q}_w$ defines a continuous linear functional.

\subsubsection*{Evaluation functionals $\langle \ppKer_{\xi}, \cdot \rangle_w$ span $\mathcal{Q}_w$} Let  $\mathcal{H} \subset \mathcal{Q}_w$ be the RKHS spanned by the kernel sections $\ppKer_{\xi}$. To show that $\mathcal{H} = \mathcal{Q}_w$ we claim that $\mathcal{H}^{\perp} = \{0 \}$ inside $\mathcal{Q}_w$. Letting $f\in \mathcal{Q}_w$  with $f\perp \ppKer_{\xi},\, \forall \xi \in \ppBall$, we have $f(\xi)=0$ for all $\xi\in \ppBall$ (see \cite[Proposition~4.4]{mujica2010}). Since $f$ is holomorphic on $\ppBall$, we have $f=0$, proving that $\mathcal{H}=\mathcal{Q}_w$. \qed

\subsection*{Part (ii). $\mathcal{Q}_w \cong \ppAlg_{\ppWFunc}$}

Let $\tilde{\ppKer} (\xi,\eta) = \ppKer_{\xi}(\eta) $ and $\ppKer(\xi, \eta) = \ppWFunc (\langle  \xi,\eta\rangle_{\ppfxHil_{\mu}})$ with $\Lambda(z) = \sum_n \lambda_n z^n$. We claim that $\tilde{\ppKer}(\xi, \eta) = \ppKer(\xi, \eta)$ if $\lambda_n^{-1} = n! w(n)^2 $. Indeed,
\begin{align*}
\tilde{\ppKer}(\xi,\eta) &= \sum_{I \in \mathcal I} \frac{1}{ w(|I|)^2 I!} \frac{\overline{z}^I(\xi) z^I(\eta)}{k^{2I}} \\
&= \sum_{n=0}^{\infty} \frac{1}{ w(n)^2 n!} \sum_{|I|=n} \frac{n!}{I!} \frac{\overline{z}^I(\xi) z^I(\eta)}{k^{2I}}\\ 
&= \sum_{n=0}^{\infty} \frac{1}{ w(n)^2 n!} \langle \xi, \ppfxCov^{-1} \eta \rangle_{\ppfxHil}^n \\
&= \ppKer(\xi, \eta),
\end{align*}
so $\mathcal{Q}_w$ and $\ppAlg_{\ppWFunc}$ are RKHSs of holomorphic functions on $\ppBall$ with the same kernel $\ppKer = \tilde{\ppKer}$. It follows that $\mathcal{Q}_w$ coincides with $\ppAlg_{\ppWFunc}$. \qed

 \subsection*{Part (iii). RKHA on $\ppBall$}
We have already shown that $\ppAlg_{\ppWFunc}\cong \mathcal{Q}_w$ is an RKHS.  Now, we identify $\mathcal{Q}_w$ with the weighted Fock space $F_w(\mathbb{H}_{\mu}^*)$.
Using an argument identical to the proof of Proposition~\ref{prop:FockL2}, we first extend $\theta(\xi^{\vee I}) = z^I$ to a map from the tensor algebra of $\ppfxHil_{\mu}^*$ to polynomials in $\mathcal{Q}_w$. This gives an isometry that preserves the multiplication structure. We then take the completions relative to the inner product of $F_w(\mathbb{H}_{\mu}^*)$ and $\mathcal{Q}_w$, respectively, to get an isometric isomorphism $\theta: F_w(\mathbb{H}_{\mu}^*) \rightarrow \mathcal{Q}_w$.  For a given weight $w$ (or equivalently $\Lambda$), the map $\theta$ is an isometry and preserves the multiplication structure. In particular, if $w$ is chosen such that  $F_w(\mathbb{H}_{\mu}^*)$ is an algebra then $\mathcal{Q}_w$ is also an algebra of functions under pointwise multiplication and the following diagram commutes,
    $$
\begin{matrix}

  F_w(\mathbb{H}^*_{\mu}) \times F_w(\mathbb{H}^*_{\mu})                    & \arrow{r}{\vee}   & F_w(\mathbb{H}^*_{\mu})                    \cr
  \arrow{d}{\theta \times \theta} &                      & \arrow{d}{\theta} \cr
   \mathcal{Q}_w \times \mathcal{Q}_w        & \arrow{r}{m} & \mathcal{Q}_w            \cr
\end{matrix}
$$
This shows that if the weight are chosen so that $F_{w}(\ppfxHil_{\mu}^*)$ is an RKHA, the same is true for $\mathcal{Q}_w$ and hence $\ppAlg_{\ppWFunc}$. Using \cite[Theorem~9]{GiannakisEtAl25}, $F_w(\mathbb{H}_{\mu}^*)$ is closed under multiplication, and an RKHA, when $\lambda_n = 1/\left(n!w(n)^2\right)$ is subconvolutive. Recall that in this paper we define the weighted Fock space as $F_w(H)=\bigoplus_{n=0}^{\infty} w(n)^2 n! H^{\vee n}$, while \cite{GiannakisEtAl25} uses $\mathbf{F}_{\tilde{w}}(H)=\bigoplus_{n=0}^{\infty} \tilde{w}(n)^2 H^{\vee n}$, implying that $w(n)^2 = \tilde{w}(n)^2/n!$. But $\Lambda \in \ppCone^{W}_{ess}$ means $\Lambda^2 \leq C \Lambda$, where $\leq$ is induced from the positive cone. Translating $\Lambda^2 \leq C \Lambda$ in terms of the Fourier coefficients $\lambda_n$ this implies that $\lambda_n$ is a subconvolutive sequence. This finishes the proof. \qed

\section{Further observations}
\label{sec:further}

\subsection{Integral representation and Gelfand triple}\label{sec:integral-rep}
We discuss a characterization of the RKHAs $\ppAlg_{\Lambda}$ using kernel integral operators on $\mathcal{H}L^2 (\ppfxHil, \mu)$ induced by the kernel $\ppKer$ and the Gaussian measure $\mu$.
This construction can be viewed as an infinite-dimensional analogue of the use of singular kernels to define (fractional) Sobolev spaces. See, for example, \cite{Stein1971SingularIA}. 
Note that elements of $\mathcal{H}L^2 (\ppfxHil, \mu)$ are holomorphic on $\ppfxHil$ while the kernel $\ppKer$ is only defined on a disk contained in $\ppfxHil_{\mu} \subset \ppfxHil$.

Let $\mathcal{D}_N \subset \mathcal{H}L^2 (\ppfxHil, \mu) $ be the space of continuous polynomials of degree at most $N$, and define $\Lambda_N \colon \mathbb C \to \mathbb C$ by $\ppWFunc_N(z) = \sum_{n=0}^{N} \lambda_n z^n$ ($\ppWFunc_N$ approaching $\ppWFunc$ as $N\rightarrow \infty$). The positive-definite kernel $\ppKer_{\ppWFunc_N}(\xi,\eta) = \Lambda_N \left( \langle \xi, \ppfxCov^{-1} \eta \rangle \right) $, defined on $\mathbb H_\mu \times \mathbb H_\mu$, induces a linear operator $T_{N}\colon \mathcal{D}_N \rightarrow \mathcal{D}_N$ as    
\begin{equation*}
    T_N f(\eta) = \int_{\ppfxHil} \ppKer_{\ppWFunc_N}(\eta,\xi)\, f(\xi) \, d\mu(\xi), \quad \eta \in \ppfxHil_{\mu} \subset \ppfxHil.  
\end{equation*}
Note that $\mathcal{D}_N$ is spanned by monomials $z^I, |I|\leq N$. Moreover, we have,
\begin{equation*}
    \ppKer_{\ppWFunc_N}(\xi, \eta) = \sum_{|I|\leq N} \frac{1}{ w(|I|)^2 I!} \frac{\overline{z}^I(\xi) z^I(\eta)}{k^{2I}}.
\end{equation*}
leading to
\begin{align}
\nonumber \left( T_N z^J \right) (\eta) &= \int_{\ppfxHil}  \sum_{|I|\leq N} \frac{1}{ w(|I|)^2 I!} \frac{z^I(\eta) \overline{z}^I(\xi)}{k^{2I}}\, z^J(\xi) \, d\mu(\xi)    \\ 
\nonumber &= \sum_{|I|\leq N} \frac{1}{ w(|I|)^2 I!} \frac{z^I(\eta) }{k^{2I}} \int_{\ppfxHil}   \overline{z}^I(\xi) z^J(\xi) \, d\mu(\xi)    \\ 
\label{eq:TN_monomials} &= \frac{1}{ w(|J|)^2} z^J(\eta).
\end{align}
In other words, $T_N$ acts diagonally with respect to the monomial basis $z^I$.

Next, we equip $\mathcal{D}_N$ with the inner product $\langle \cdot, \cdot \rangle_N$ induced from the inner product of $L^2(\ppfxHil, \mu)$, making it a Hilbert space. We also define $\mathcal{D}_N'$ to be the same vector space as $\mathcal{D}_N$ but equipped with the inner product  $\langle\langle f, g \rangle\rangle_N = \langle f, T_N^{-1} g \rangle_{L^2(\ppfxHil, \mu)}$. 
Note that \eqref{eq:TN_monomials} implies that $ \|  z^J \|_{\mathcal{D}'_N} = \| T_N^{-1/2} z^J \|_{\mathcal{D}_N}$, and hence $T_N^{1/2}$ defines an isometric isomorphism from $\mathcal{D}_N$ to $\mathcal{D}'_N$. It is also not hard to see that $\ppAlg_{\Lambda}$ is the direct limit of Hilbert spaces $\mathcal{D}'_N$ as $N\rightarrow \infty$, and hence $T^{1/2}_N$ extends to an isometric isomorphism $T: \mathcal{H}L^2 (\ppfxHil, \mu) \rightarrow \ppAlg_{\Lambda}$.

If, instead, we consider $T$ as an operator mapping $ \mathcal{H}L^2 (\ppfxHil, \mu)$ to itself then it is unbounded and densely defined with domain $\bigcup_N \mathcal D_N$.   
In our case of $\lambda(n)=e^{-\tau n^p}$ (or equivalently $w(n) = \sqrt{e^{\tau n^p}/n!}$), we have  $\mathcal{H}L^2 (\ppfxHil, \mu) \subset  \ppAlg_{\Lambda}$. Furthermore, if we define 
$$
\hat \ppAlg_{\Lambda} = \left\{f: \ppBall \rightarrow \mathbb{C}:\; f(\xi) = \sum_{I \in \mathcal I} a_I z^I(\xi):\;  \sum_{I \in \mathcal I}  \frac{|a_I|^2}{w(|I|)^2} I! k^{2I} < \infty \right\}
$$
then $\hat \ppAlg_{\Lambda} \subset \mathcal{H}L^2 (\ppfxHil, \mu) \subset  \ppAlg_{\Lambda}$ forms a Gelfand triple. Such a triple is the basic ingredient in the construction of generalized Brownian
functionals; see \cite{Kuo1983BrownianFA,Hida1976ANALYSISOB,STREIT198455}.

\subsection{Twisted canonical commutation relations}
The CCR algebras generated by creation and annihilation operators plays a central role in the theory of Fock spaces \cite{MR1770752}. In the Segal--Bargmann representation of the CCR algebra for one variable the annihilation and creation operators, $a$ and $a^\dagger$ are represented by unbounded multiplication and differentiation operators on holomorphic functions, respectively,
\begin{equation*}
    (a f)(s) = f'(s), \quad (a^\dagger f)(s) = s f(s), \quad s\in \mathbb{C}.
\end{equation*}
Since $\ppAlg_{\ppWFunc}$ can be identified with a weighted Fock space, it is natural to ask whether analogous relations hold in this weighted setting. In this section, we introduce the creation and annihilation operators on $\ppAlg_{\ppWFunc}$ and show that they are bounded operators satisfying a twisted form of the CCRs.  

For $j \in \mathbb N$, we define the creation operator $a_j^{\dagger}$ via its action on monomials, 
$$a_j^{\dagger} z^I = \frac{1}{k_j} z^{I+j},$$
where $I+j = (i_1, \dots, i_{j-1}, i_{j}+1,i_{j+1}, \dots )$. More specifically, $a_j^{\dagger}$ is the multiplication operator by $\frac{z_j}{k_j}$ where $z_j: \ppfxHil \rightarrow \mathbb{C}$ is the linear function $z_j(\xi) = \langle \xi_j, \xi \rangle$. Recall that $\| z^I \|_{\ppAlg_{\ppWFunc}} = w(|I|) \sqrt{I!} k^{I}$, in particular $\| z_j \|_{\ppAlg_{\ppWFunc}} = k_j$, and hence $\frac{z_j}{k_j}$ is the normalized coordinate function.  Since $\ppAlg_{\ppWFunc}$ is an RKHA, it is evident from this definition that $a_j^{\dagger}$ is a bounded operator. The annihilation operator, $a_j$, is then defined as the adjoint of $a_j^{\dagger}$ with respect to the inner product of $\ppAlg_{\ppWFunc}$. Similarly, we define the operator $D$ by $Dz^I = w(n)^2 z^I$ where $n=|I|$. The twisted bracket is then defined by  
$$[A, B]_D = [\operatorname{Ad}_D(A), B], \quad \operatorname{Ad}_D(A) = D A D^{-1}.$$
Note that for our weights $w$ in Theorem~\ref{thm:main}, $D$ is bounded but $D^{-1}$ is unbounded. The twisted bracket $[A, B]_D$ is then well-defined for all bounded operators $A$ on $\mathcal A_\Lambda$ and all bounded operators $B$ on $ \mathcal A_\Lambda$ whose range is contained in the domain of $D^{-1}$.

Next, recall that the Gelfond--Leontiev differential operator \cite{Gelfond1951}, $\mathcal D_\phi$, associated with an entire function $\phi(s) = \sum_{m=0}^\infty \phi_m s^m$ with $\phi_m \neq 0$ is defined as the operator acting on analytic functions $f(s) = \sum_{m=0}^\infty a_m s^m$ as
\begin{equation*}
    \mathcal D_\phi(s) = \sum_{m=1}^\infty a_m \frac{\phi_{m-1}}{\phi_m} s^{m-1}.
\end{equation*}
The choice $\phi(s) = e^s$ then yields the standard derivative $\mathcal D_{e^s}f(s) = f'(s)$.
The following proposition shows that the annihilation operators $a_j$ on $\ppAlg_{\ppWFunc}$ are Gelfond--Leontiev differential operators  satisfying, together with $a_j^\dagger$, twisted CCR relations. See \cite[Section~3]{Alpay2022} and \cite{Gelfond1951} for the one-dimensional case). 

\begin{prop}
    The action of annihilation operator $a_j\colon \mathcal A_\Lambda \to \mathcal A_\Lambda$ on monomials is given by 
    \begin{equation*}
        a_j z^I = k_j\, i_j \frac{\lambda_{n-1}}{\lambda_n} z^{I-j}, \quad n=|I| \geq 1,
    \end{equation*}
    and $a_j z^0 = 0$. In particular, $a_j$ is a Gelfond--Leontiev differential operator associated with $\ppWFunc(z) = \sum_{m=0}^\infty \lambda_m z^m$. Moreover, $a_j$ and $a_j^\dagger$ satisfy the twisted CCRs
        $$
    [a_j, a_j^{\dagger}]_D = \operatorname{Id}.
    $$
\end{prop}

\begin{proof}
    First, we verify the action of $a_j$ on monomials $z^I$, $I \neq 0$, using 
\begin{align*}
    a_j z^I &= \sum_{J \in \mathcal I} \left\langle a_j z^I, \frac{z^J}{\|z^J\|_{\ppAlg_\Lambda}} \right\rangle_{\ppAlg_\Lambda}\, \frac{z^J}{\|z^J\|_{\ppAlg_\Lambda}}  
    = \sum_{J \in \mathcal I} \left\langle z^I, a_j^{\dagger} z^J \right\rangle_{\ppAlg_\Lambda}\, \frac{z^J}{\|z^J\|^2_{\ppAlg_\Lambda}} \\
    &= \sum_{J \in \mathcal I} \left\langle z^I, z^{J+j}  \right\rangle_{\ppAlg_\Lambda}\, \frac{z^J}{k_j \|z^J\|^2_{\ppAlg_\Lambda}}
    = \frac{\|z^I\|^2_{\ppAlg_\Lambda}}{\|z^{I-j}\|^2_{\ppAlg_\Lambda} k_j } z^{I-j}\\
    &= \frac{I! w(n)^2 k_j }{(I-j)! w(n-1)^2} z^{I-j} = \frac{i_j\; w(n)^2 k_j}{w(n-1)^2} z^{I-j}.
\end{align*}
    where $I-j = (i_1, \dots, i_{j-1}, i_{j}-1,i_{j+1}, \dots )$. The case $a_j z^0 = 0$ follows similarly, using the fact that the constant function $z^0$ is orthogonal to every linear function $z_j$. Next, we show that $[a_j, a_j^{\dagger}]_D z^I = z^I$ for the monomial basis $z^I, I\in \mathcal{I}$. To that end, note that
    \begin{align*}
        D a_j D^{-1}  a_j^{\dagger} z^I = D a_j \left( \frac{1}{ w(n+1)^2 k_j} z^{I+j} \right)  
        &= \frac{1}{w(n+1)^2} D  \left( \frac{(i_j +1) w(n+1)^2}{w(n)^2} z^I \right)\\
        &= (i_j + 1) z^I.
    \end{align*}    
    Similarly, we have $a_j^{\dagger} D a_j D^{-1} z^I = i_j z^I$, and hence
    \begin{equation*}
        [a_j, a_j^{\dagger}]_D z^I = (i_j + 1) z^I -  i_j z^I = z^I. \qedhere
    \end{equation*}
\end{proof}

In Fock space theory, the standard coherent states are defined as eigenvectors of annihilation operators. They and their generalizations find applications in in many areas, including quantum theory, signal processing, and wavelet theory; e.g., \cite{TwarequeEtAl95}. In the context of the RKHA $\mathcal A_\Lambda$, it follows from the fact that $a_j^\dagger$ is a multiplication operator that every kernel section $\varepsilon_\eta = \mathcal K_\Lambda(\eta, \cdot)$ is an eigenvector of $a_j$ with corresponding eigenvalue $\alpha_j(\eta) := \langle \xi_j, \eta \rangle_{\mathbb H_\mu}$. More specifically, for any $f\in \ppAlg_\Lambda$, 
\begin{equation*}
    \langle a_j \varepsilon_{\eta}, f \rangle_{\ppAlg_\Lambda} = \langle  \varepsilon_{\eta}, a_j^{\dagger} f \rangle_{\ppAlg_\Lambda} = \langle \varepsilon_{\eta}, z_j f \rangle_{\ppAlg_\Lambda} = z_j(\eta) f(\eta) = \langle \alpha_j(\eta) \varepsilon_\eta,  f \rangle_{\ppAlg_\Lambda},
\end{equation*} 
and hence 
 \begin{equation*}
    a_j \varepsilon_\eta = \alpha_j(\eta) \varepsilon_\eta.
\end{equation*}
On the basis of the above, we interpret the $\varepsilon_\eta$ as coherent states associated with the twisted CCR algebra generated by the $a_j$ and $a_j^\dagger$.

\section*{Acknowledgments}

Dimitrios Giannakis acknowledges support from the U.S.\ Department of Defense, Basic Research Office, under Vannevar Bush Faculty Fellowship grant N00014-21-1-2946 and the U.S.\ Office of Naval Research under MURI grant N00014-19-1-242.
Mohammad Javad Latifi Jebelli and Michael Montgomery were supported as postdoctoral fellows from these grants.

\begin{table}[h]
 \small
\centering
\renewcommand{\arraystretch}{1.15}
\caption{\label{tab:symbols}Main symbols used in the paper.}
\begin{tabular}{p{0.28\textwidth} p{0.66\textwidth}}
\hline
\textbf{Symbol} & \textbf{Meaning} \\
\hline

$\ppfxHil$ & Infinite-dimensional separable Hilbert space. \\

$\ppfxCov \colon \ppfxHil\to\ppfxHil$ & Positive-definite trace-class operator. \\

$\xi_j$ & Orthonormal eigenbasis of $\ppfxCov$ with $\ppfxCov\xi_j=k_j^2\xi_j$. \\

$\ppfxHil_\mu$ & Cameron--Martin Hilbert space associated with $\mu$. \\

$\mu$ & Gaussian measure on $\ppfxHil$ with covariance  $\ppfxCov$. \\

$\ppBall$ & The ball
$\ppBall=\{\xi\in\ppfxHil:\|\ppfxCov^{-1/2}\xi\|<1\}$. \\

$\Delta$ & Unit disk in $\mathbb C$: $\Delta=\{z\in\mathbb C:|z|<1\}$. \\

$W(\Delta)$ & Wiener algebra \\

$\ppCone^W$ & Positive cone in $W(\Delta)$ with coefficients $\lambda_n\ge0$. \\

$\ppCone^W_{\mathrm{ess}}$ & Essential elements of $\ppCone^W$: $a^2\le Ca$ for some $C>0$. \\ 

$\ppWFunc$ & A function in $W(\Delta)$, $\Lambda(z)=\sum_{n\ge0}\lambda_n z^n$. \\

$\ppKer_{\ppWFunc}$ & Kernel on $\bar\ppBall\times\bar\ppBall$:
$\ppKer_{\ppWFunc}(\xi,\eta)=\ppWFunc(\langle \xi,\ppfxCov^{-1}\eta\rangle)$. \\

$\mathcal A_\Lambda$ & RKHS on $\bar\ppBall$ induced by $\ppKer_{\ppWFunc}$ .\\

$\Delta\colon \ppAlg\to \ppAlg\otimes\ppAlg$ & Comultiplication operator for an RKHA: $\Delta^*(f\otimes g)=fg$. \\

$\mathcal I$ & Nonnegative multi-indices $I=(i_1,i_2,\dots)$ with $\sum_j i_j < \infty$. \\

$|I|$, $I!$ & Length and factorial: $|I|=\sum_r i_r$, $I!=\prod_r i_r!$. \\

$z_j(\xi)$ & Coordinate function on $\ppfxHil$: $z_j(\xi)=\langle \xi,\xi_j\rangle$. \\

$z^I(\xi)$ & Monomial: $z^I(\xi)=\prod_{j} z_j(\xi)^{i_j}$. \\

$L^2(\ppfxHil,\mu)$ & Square-integrable functions w.r.t.\ $\mu$. \\

$\mathcal{H}L^2(\ppfxHil,\mu)$ & Segal--Bargmann space (holomorphic part of $L^2(\ppfxHil,\mu)$). \\

$\ppHil^{\otimes n}$, $\ppHil^{\vee n}$ & $n$-fold tensor product and its symmetric subspace. \\

$F(\ppHil)$ & Symmetric Fock space.\\

$\vee$ & Symmetric tensor product. \\

$\xi^{\vee I}$ & Symmetric tensor monomial corresponding to multi-index $I$:
$\xi^{\vee I}=\xi_1^{\vee i_1}\vee\xi_2^{\vee i_2}\vee\cdots$. \\

$\theta$ & Segal--It\^o type map $F(\ppfxHil_\mu^*)\to \mathcal{H}L^2(\ppfxHil,\mu)$. \\

$w\colon\mathbb N\to\mathbb R_{>0}$ & Weight related to $\Lambda$  via $\lambda_n^{-1}=n!\,w(n)^2$. \\

$\mathcal Q_w$ & Completion of polynomials using weighted inner product. \\

$a_j^\dagger$, $a_j$ & Creation and annihilation operators. \\

$D$ & Diagonal operator $Dz^I=w(|I|)^2 z^I$. \\

$[A,B]_D$ & Twisted commutator $[A,B]_D=[DAD^{-1},B]$. \\

\hline
\end{tabular}
\end{table}

\bibliography{bibliography,bibliography_dg}

@article{Alpay2022,
  author        = {Alpay, D. and Cerejeiras, P. and K\"{a}hler, U. and Kling, T.},
  title         = {{C}ommutators on {F}ock spaces},
  journal       = {J. Math. Phys.},
  volume        = {64},
  number        = {4},
  pages         = {042102},
  year          = {2023},
  issn          = {0022-2488},
  doi           = {10.1063/5.0080723},
}

@article{Gelfond1951,
  author        = {Gelfond, A. O. and Leontev, A. F.},
  title         = {On a generalization of {F}ourier series},
  journal       = {Mat. Sbornik N.S.},
  fjournal      = {Mat. Sbornik N.S.},
  volume        = {29/71},
  year          = {1951},
  pages         = {477--500},
  mrclass       = {30.0x},
  mrnumber      = {45812},
  mrreviewer    = {I.\ M.\ Sheffer},
}

@book{Segal1963,
  title         = {Mathematical problems of relativistic physics},
  author        = {Segal, I. E.},
  volume        = {2},
  year          = {1963},
  publisher     = {American Mathematical Society},
  address       = {Providence, RI},
}

@article{Trieu2017,
  issn          = {03794024, 18417744},
  url           = {https://www.jstor.org/stable/26432293},
  author        = {T. Le},
  journal       = {J. Oper. Theory},
  number        = {1},
  pages         = {pp. 135--158},
  publisher     = {Theta Foundation},
  title         = {Composition Operators Between {S}egal–{B}argmann Spaces},
  urldate       = {2026-04-03},
  volume        = {78},
  year          = {2017},
}

@article{hartz2025weighted,
  title         = {Weighted composition operators on {H}ilbert function spaces on the ball},
  author        = {Hartz, M. and Tornes, M.},
  journal       = {Results Math.},
  volume        = {80},
  number        = {6},
  pages         = {180},
  year          = {2025},
  publisher     = {Springer},
}

@book{mujica2010,
  title         = {Complex Analysis in {B}anach Spaces},
  author        = {Mujica, J.},
  isbn          = {9780486474663},
  lccn          = {2009031304},
  series        = {Dover Books on Mathematics},
  url           = {https://books.google.com/books?id=nQg7BAAAQBAJ},
  year          = {2010},
  publisher     = {Dover Publications},
}

@book{Bogachev1998,
  title         = {Gaussian Measures},
  author        = {V. I. Bogachev},
  url           = {https://api.semanticscholar.org/CorpusID:264752010},
  year          = {1998},
  publisher     = {American Mathematical Soc},
}

@book{Simon1974,
  url           = {http://www.jstor.org/stable/j.ctt13x16st},
  author        = {B. Simon},
  publisher     = {Princeton University Press},
  title         = {P$(\Phi)_2$ {E}uclidean (Quantum) Field Theory},
  urldate       = {2025-02-26},
  year          = {1974},
}

@incollection{MR1770752,
  author        = {Hall, B. C.},
  title         = {Holomorphic methods in analysis and mathematical physics},
  booktitle     = {First {S}ummer {S}chool in {A}nalysis and {M}athematical {P}hysics ({C}uernavaca {M}orelos, 1998)},
  series        = {Contemp. Math.},
  volume        = {260},
  pages         = {1--59},
  publisher     = {Amer. Math. Soc., Providence, RI},
  year          = {2000},
  isbn          = {0-8218-2115-6},
  mrclass       = {81s05 (22e30 32a36 46e22 46n50 81s30)},
  mrnumber      = {1770752},
  mrreviewer    = {Gerald\ B.\ Folland},
  doi           = {10.1090/conm/260/04156},
  url           = {https://doi.org/10.1090/conm/260/04156},
}

@article{Kuo1983BrownianFA,
  title         = {Brownian functionals and applications},
  author        = {H. H. Kuo},
  journal       = {Acta Applicandae Mathematica},
  year          = {1983},
  volume        = {1},
  pages         = {175--188},
  url           = {https://api.semanticscholar.org/CorpusID:122267558},
}

@inproceedings{Hida1976ANALYSISOB,
  title         = {{A}nalysis Of {B}rownian {F}unctionals},
  author        = {T. Hida},
  year          = {1976},
  url           = {https://api.semanticscholar.org/CorpusID:123441258},
}

@article{STREIT198455,
  title         = {Generalized {B}rownian functionals and the {F}eynman integral},
  journal       = {Stoch. Process. Their Appl.},
  volume        = {16},
  number        = {1},
  pages         = {55--69},
  year          = {1984},
  issn          = {0304-4149},
  doi           = {https://doi.org/10.1016/0304-4149(84)90175-3},
  url           = {https://www.sciencedirect.com/science/article/pii/0304414984901753},
  author        = {L. Streit and T. Hida},
}

@inproceedings{Stein1971SingularIA,
  title         = {Singular Integrals and Differentiability Properties of Functions},
  author        = {E. M. Stein},
  year          = {1971},
  url           = {https://api.semanticscholar.org/CorpusID:115462014},
}

@article{Arveson72,
  author        = {Arveson, W.},
  title         = {Subalgebras of {$C^*$}-Algebras {II}},
  journal       = {Acta Math.},
  volume        = {128},
  number        = {3--4},
  pages         = {271--308},
  year          = {1972},
  doi           = {10.1007/bf02392166},
}

@article{Babuska68a,
  author        = {Babu{\v s}ka, I.},
  title         = {{\"U}ber Universal Optimale {Q}uadraturformeln, {T}eil {I}},
  journal       = {Appl. Math.},
  year          = {1968},
  volume        = {13},
  pages         = {304--338},
}

@article{Babuska68b,
  author        = {Babu{\v s}ka, I.},
  title         = {{\"U}ber Universal Optimale {Q}uadraturformeln, {T}eil {II}},
  journal       = {Appl. Math.},
  year          = {1968},
  volume        = {13},
  pages         = {388--404},
}

@incollection{Beurling38,
  author        = {Beurling, A.},
  year          = {1938},
  title         = {Sur les int{\'e}grales de {F}ourier absolument convergentes et leur application {\`a} une transformation fonctionnelle},
  booktitle     = {Arne {B}eurling: Collected Works, Vol. II},
  editor        = {Carleson, L. and Malliavin, P. and Neuberger, J. and Wermer, J.},
  publisher     = {Birkh{\"a}user},
  address       = {Boston},
}

@article{DasGiannakis23,
  author        = {Das, S. and Giannakis, D.},
  title         = {On Harmonic {H}ilbert Spaces on Compact Abelian Groups},
  year          = {2023},
  journal       = {J. Fourier Anal. Appl.},
  volume        = {29},
  number        = {1},
  pages         = {12},
  eid           = {12},
  doi           = {10.1007/s00041-023-09992-4},
}

@article{DasEtAl23,
  author        = {Das, S. and Giannakis, D. and Montgomery, M.},
  title         = {Correction to: On Harmonic {H}ilbert Spaces on Compact Abelian Groups},
  journal       = {J. Fourier Anal. Appl.},
  year          = {2023},
  volume        = {29},
  number        = {6},
  eid           = {67},
  pages         = {67},
  doi           = {10.1007/s00041-023-10043-1},
}

@article{Delvos97,
  author        = {F. J. Delvos},
  title         = {Interpolation in Harmonic {H}ilbert Spaces},
  journal       = {Math. Model. Numer. Anal.},
  year          = {1997},
  volume        = {31},
  number        = {4},
  pages         = {435--458},
  url           = {http://www.numdam.org/item/M2AN_1997__31_4_435_0/},
}

@article{Feichtinger79,
  author        = {Feichtinger, H. G.},
  title         = {Gewichtsfunktionen auf Lokalkompakten {G}ruppen},
  journal       = {{\"O}sterreich. Akad. Wiss. Math.-Natur. Kl. Sitzungsber. II},
  year          = {1979},
  volume        = {188},
  number        = {8--10},
  pages         = {451--471},
}

@article{FeichtingerEtAl07,
  author        = {Feichtinger, H. G. and Pandey, S. S. and Werther, T.},
  title         = {Minimal Norm Interpolation in Harmonic {H}ilbert Spaces and {W}iener Amalgam Spaces on Locally Compact Abelian Groups},
  journal       = {J. Math. Kyoto Univ.},
  year          = {2007},
  volume        = {47},
  number        = {1},
  pages         = {65--78},
  doi           = {10.1215/kjm/1250281068},
}

@article{GiannakisMontgomery25,
  author        = {Giannakis, D. and Montgomery, M.},
  title         = {An Algebra Structure for Reproducing Kernel {H}ilbert Spaces},
  journal       = {Banach J. Math. Anal.},
  volume        = {19},
  year          = {2025},
  eid           = {11},
  doi           = {10.1007/s43037-024-00398-y},
}

@misc{GiannakisEtAl25,
  author        = {Giannakis, D. and Latifi Jebelli, M. J. and Montgomery, M. and Pfeffer, P. and Schumacher, J. and Slawinska, J.},
  title         = {Second quantization for classical nonlinear dynamics},
  year          = {2025},
  url           = {https://arxiv.org/abs/2501.07419v1},
}

@article{Gneiting13,
  author        = {Gneiting, T.},
  title         = {Strictly and non-strictly positive definite functions on spheres},
  journal       = {Bernoulli},
  volume        = {19},
  number        = {4},
  year          = {2013},
  pages         = {1327--1349},
  doi           = {10.3150/12-bejsp06},
}

@incollection{Grochenig07,
  author        = {Gr{\"o}chenig, K.},
  title         = {Weight Functions in Time-Frequency Analysis},
  booktitle     = {Pseudodifferential Operators: Partial Differential Equations and Time-Frequency Analysis},
  publisher     = {American Mathematical Society},
  year          = {2007},
  pages         = {343--366},
  address       = {Providence},
  editor        = {Rodino, L. and others},
  volume        = {52},
  series        = {Fields Inst. Commun.},
}

@book{Kaniuth09,
  author        = {Kaniuth, E.},
  title         = {A Course in Commutative {B}anach Algebras},
  year          = {2009},
  location      = {New York},
  publisher     = {Springer Science+Media},
  series        = {Graduate Texts in Mathematics},
  volume        = {246},
}

@article{TwarequeEtAl95,
  author        = {Twareque Ali, S. and Antoine, J.-P. and Gazeau, J.-P. and Mueller, U. A.},
  title         = {Coherent States and their Generalizations: {A} Mathematical Review},
  journal       = {Rev. Math. Phys.},
  volume        = {7},
  number        = {7},
  pages         = {1013--1104},
  year          = {1995},
  doi           = {10.1142/s0129055x95000396},
}

@article{Wermer54,
  author        = {Wermer, J.},
  title         = {On a Class of Normed Rings},
  journal       = {Ark. Mat.},
  year          = {1954},
  volume        = {2},
  number        = {6},
  pages         = {537--551},
  doi           = {10.1007/bf02591228},
}

\end{document}